\documentclass[preprint,12pt]{elsarticle}

\usepackage[latin9]{inputenc}
\usepackage{color}
\usepackage{verbatim}
\usepackage{amsmath}
\usepackage{stmaryrd}
\usepackage{eepic}
\usepackage{epsfig}
\usepackage{amssymb,amsbsy}
\usepackage{amsthm}
\usepackage{amscd}
\usepackage{comment}
\usepackage{graphicx}
\usepackage{esint}
\usepackage{amscd}
\usepackage[normalem]{ulem}
\usepackage[lined,boxed,commentsnumbered]{algorithm2e}
\usepackage{caption}
\usepackage{subcaption}

\newtheorem{theo}{Theorem}[section]

\newtheorem{lemma}[theo]{Lemma}
\newtheorem{proposition}[theo]{Proposition}

\newtheorem{assumption}{Assumption}[section]
\theoremstyle{definition}
\newtheorem{remark}[theo]{Remark}
\theoremstyle{plain}

\def\mbd{\mbox{d}}

\usepackage{amssymb}

\journal{Journal of Computational Physics}

\begin{document}

\begin{frontmatter}



\title{Hybrid framework for the simulation of stochastic chemical kinetics}


\author[aduncan]{Andrew Duncan\corref{cor1}}
\ead{a.duncan@imperial.ac.uk}

\author[rerban]{Radek Erban}
\ead{erban@maths.ox.ac.uk}

\author[kzygalakis]{Konstantinos Zygalakis}
\ead{k.zygalakis@soton.ac.uk}

\cortext[cor1]{Corresponding author}
\address[aduncan]{Department of Mathematics, Imperial College,
South Kensington Campus, London, SW7 2AZ, United Kingdom}
\address[rerban]{Mathematical Institute, University of Oxford,
Radcliffe Observatory Quarter, Woodstock Road, Oxford, OX2 6GG,
United Kingdom}
\address[kzygalakis]{School of Mathematics, University of Edinburgh,
Peter Guthrie Tait Road, Edinburgh, EH9 3FD, United Kingdom}
\begin{abstract}
\noindent Stochasticity plays a fundamental role in various 
biochemical processes, such as cell regulatory networks and enzyme 
cascades. Isothermal, well-mixed systems can be modelled as Markov 
processes, typically simulated using the Gillespie Stochastic 
Simulation Algorithm (SSA)~\cite{Gillespie:1977:ESS}.  
While  easy to implement and exact, the computational cost of 
using the Gillespie SSA to simulate such systems can become prohibitive 
as the frequency of reaction events increases.  This has motivated 
numerous coarse-grained schemes, where the ``fast'' reactions 
are approximated either using Langevin dynamics or deterministically.  
While such approaches provide a good approximation when all 
reactants are abundant, the approximation breaks down when one 
or more species exist only in small concentrations and 
the fluctuations arising from the discrete nature of the 
reactions becomes significant. This is particularly problematic 
when using such methods to compute statistics of extinction times 
for chemical species, as well as simulating non-equilibrium 
systems such as cell-cycle models in which a single species 
can cycle between abundance and scarcity. In this paper, 
a hybrid jump-diffusion model for simulating well-mixed stochastic 
kinetics is derived. It acts as a bridge between the Gillespie SSA 
and the chemical Langevin equation. For low reactant reactions 
the underlying behaviour is purely discrete, while purely diffusive 
when the concentrations of all species is large, with the two 
different behaviours coexisting in the intermediate region. 
A bound on the weak error in the classical large volume scaling 
limit is obtained, and three different numerical discretizations 
of the jump-diffusion model are described.  The benefits of such 
a formalism are illustrated using computational examples.
\end{abstract}


\end{frontmatter}


\section{Introduction}
\label{sec:intro}
\noindent
Biochemical systems with small numbers of interacting components  
have increasingly been studied in the recent years. Examples include 
the phage $\lambda$ lysis-lysogeny decision 
circuit~\cite{Arkin:1998:SKA}, circadian rhythms~\cite{Villar:2002:MNR} 
and cell cycle \cite{Kar:2009:ERN}. It is this small number of interacting
components that makes the appropriate mathematical framework for 
describing these systems a stochastic one. In particular, the 
kinetics of the different species is accurately described, under 
appropriate assumptions, by a  continuous-time discrete-space 
Markov chain. The theory of stochastic  
processes~\cite{Gardiner:1985:HSM,vanKampen:2007:SPP} allows the 
association of the Markov chain with an underlying  master equation, 
which is a set of ordinary differential equations (ODEs), possible 
of infinite dimensions, that describe, at each point in time, the 
probability density of all the different possible states of the system. 
In the context of biochemical systems this equation is known as the
\emph{chemical master equation} (CME).

The high dimensionality of the CME makes it intractable to solve 
in practice. In particular, with the exception of some very simple 
chemical systems~\cite{MR2272420} analytic solutions of the CME 
are not available. One way to deal with this issue is to resort 
to stochastic simulation of the underlying Markov chain. 
The stochastic simulation algorithm (SSA) developed by 
Gillespie~\cite{Gillespie:1977:ESS} exactly simulates trajectories 
of the CME as the system evolves in time. The main idea behind 
this algorithm, is that at each time point, one samples a waiting 
time to the next reaction from an appropriate exponential 
distribution, while another draw of a random variable is then 
used to decide which of the possible reactions will actually 
occur. For suitable classes of chemically reacting
systems, one can sometimes use exact algorithms which, although 
equivalent to the Gillespie SSA are less computationally
intensive. Examples include the Gibson-Bruck Next Reaction 
Method~\cite{Gibson:2000:EES} and the Optimized Direct 
Method~\cite{Cao:2004:EFS}. These algorithms can be further
accelerated by using parallel computing, for example, 
on Graphics Processing Units~\cite{Klingbeil:2010:FTT,Klingbeil:2010:SPS}.

All the methods described above can only go so far in terms 
of speeding up the simulations, since even with the all 
possible speed ups running the SSA can be computationally 
intensive for realistic problems. One approach to alleviate 
the computational cost is to employ different approximations 
on the level of the description of the chemical system. For example, 
in the limit of large molecular populations, the waiting time becomes, 
on average, very small and under the law of mass action the time 
evolution of the kinetics is described by a system of ODEs. 
This system is known as the \emph{reaction rate equation} which 
describes, approximately, the time evolution of the mean of the 
evolving Markov chain. An intermediate regime between the SSA 
and the reaction rate equation is the one where stochasticity 
is still important, but there exist a sufficient number of 
molecules to describe the evolving kinetics by a continuous model.  
This regime is called the chemical Langevin equation 
(CLE)~\cite{vanKampen:2007:SPP,gillespie2000chemical}, which is 
an It\^{o} stochastic differential equation (SDE) driven by 
a multidimensional Wiener process. In this case the corresponding 
master equation for the CLE is called the chemical Fokker-Planck 
equation (CFPE) which is a $N$-dimensional parabolic partial 
differential equation, where $N$ is the number of the different 
chemical species present in the system. 

The fact that stochasticity is still present in the  description 
of the chemical system, combined with the fact that the underlying 
CFPE is more amenable to rigorous analysis than the CME, has made 
the CLE equation a very popular regime used in 
applications~\cite{salis,SNIPER,YK09}. However, while there are 
benefits to working with the CLE/CFPE, this approximation is only 
valid in the limit of large system volume and provides  
poor approximations for systems possessing one or more chemical
species with low copy numbers. Furthermore, 
unlike the SSA/CME which ensures that there is always a positive 
(or zero) number of molecules in the system, the CFPE and CLE can 
give rise to negative concentrations, so that the chemical species 
can attain negative copy numbers. This can have serious mathematical
implications, since the CFPE equation might break down completely, 
due to regions in which the diffusion tensor is no longer positive 
definite, which makes the underlying problem ill-posed. 
On the level of the CFPE, one way to deal with such positivity 
issues is to truncate the domain and artificially impose no 
flux-boundary conditions along the 
domain boundary~\cite{cotter2013adaptive,ferm2004adaptive,sjoberg2009fokker,
SNIPER,grima2011accurate,cottererror}, which will have a negligible effect 
on the solution when it is concentrated far away from the boundary.  
When all chemical species exist in sufficiently high concentration, 
Dirichlet boundary conditions can also be used if one solves the 
stationary CFPE as an eigenvalue problem~\cite{Liao2015}.  
However, as shown in~\cite{duncan2015noise}, these artificial 
boundary conditions can result in significant approximation errors 
when the solution is concentrated near the boundary. Other alternatives 
have been proposed to overcome the behaviour of the CLE close to 
the boundary, either by suppressing reaction channels which may 
cause negativity near the boundary~\cite{dana2011physically}, 
or by extending the domain of the process to allow exploration 
in the complex domain~\cite{schnoerr2014complex}.  
In the later approach the resulting process, called the Complex 
CLE will have a positive definite diffusion tensor for all time, 
thus avoiding such breakdowns entirely. However, this method does 
not accurately capture the CME behaviour near the boundary, and 
in areas where the CLE is a poor approximation to the CME, the 
corresponding Complex CLE will suffer equally.    

These issues have motivated a number of hybrid schemes which have
been obtained by treating only certain chemical species as 
continuous variables and the others as discrete~\cite{HGK15,FCZ14,Safta2015}.
By doing so, such schemes are able to benefit from the computational 
efficiency of continuum approximations while still taking into account 
discrete fluctuations when necessary. Typically such schemes involve
partitioning the reactions into ``fast'' and ``slow'' reactions, 
with the fast reactions modelled using a continuum approximation 
(CLE or the reaction rate equation), while using Markov jump process 
to simulate the discrete reactions. Chemical species which are 
affected by fast reactions are then modelled as continuous variables 
while the others are kept discrete. Since the reaction rate depends 
on the state, it is possible that some fast reactions become slow 
and vice versa. This is typically accounted for by periodically 
repartitioning the reactions. Based on this approach, a number 
of hybrid models have been proposed, such 
as~\cite{Haseltine:2002:ASC,crudu2009hybrid}, which couple 
deterministic reaction-rate equations for the fast reactions 
with Markov jump dynamics for the slow, resulting in 
a piecewise-deterministic Markov process for the entire system.  
Error estimates for such systems, in the large volume limit, were 
carried out in~\cite{jahnke2012error}.  Similar methods have 
been proposed, such as~\cite{hasenauer2014method} and more 
recently~\cite{smith2015model}. 
Other hybrid schemes~\cite{hellander2007hybrid,menz2012hybrid} 
also involve a similar partition into slow and fast species, 
however the evolution of the slow species is obtained by solving 
the CME directly, coupled to a number of reaction-rate equations 
for the fast reactions. The hybrid system is thus reduced to 
a system of ODEs.  An error analysis of these schemes was 
carried out in~\cite{jahnke2011reduced}.

In this paper, we propose a hybrid scheme which uses Langevin 
dynamics to simulate fast reactions coupled with jump/SSA dynamics 
to simulate reactions in which the discreteness cannot be discounted.  
Thus, unlike the previously proposed models, both the continuous 
and discrete parts of the model are described using a stochastic 
formulation. Moreover, our scheme does not explicitly keep track 
of fast and slow reactions, but rather, the process will perform 
Langevin dynamics in regions of abundance, jump dynamics in regions 
in where one of the involved chemical species are in small 
concentrations,  and a mixture of both in intermediate regions.  
The resulting process thus becomes a jump-diffusion process with 
Poisson distributed jumps.  The preference of jump over Langevin 
dynamics is controlled for each individual reaction by means of 
a \emph{blending} function which is chosen to take value $1$ 
in regions of low concentration, $0$ in regions where all 
involved chemical species are abundant, and smoothly interpolates 
in between.  The choice of the blending regions will depend 
on the constants of the propensity and are generally chosen 
so that the propensity is large in the continuum region, and 
small in the discrete region. Hybrid models for chemical 
dynamics involving both jump and diffusive dynamics have been 
previously studied in various contexts. Recently, 
a method~\cite{ganguly2014jump} based on a similar coupling of 
SSA and Langevin dynamics was proposed. The authors introduce 
a partition of reactions into fast and slow reactions, applying 
the diffusion approximation to the fast reactions to obtain 
a jump-diffusion process. Based on an a-posteriori error estimator 
the algorithm periodically repartitions the species accordingly.  
By introducing the blending region our approach no longer requires 
periodic repartitioning. Other works which have considered hybrid 
schemes based on jump-diffusion dynamics 
include~\cite{angius2015approximate}. In \cite{FCE11,FCZ14} 
a hybrid scheme based on a similar domain decomposition idea was 
proposed for simulating spatially-extended stochastic-reaction 
diffusion models. In one part of the domain a SDE was used 
to simulate the position of the particles and on the other part 
a compartment-based jump process for diffusion was used. These two 
domains were separated by a sharp interface, where corrections to 
the transition probabilities at the interface were applied to 
ensure that probability mass was transferred between domains.  
While such a direct matching between continuum and discrete 
fluxes at the interface can accurately simulate systems having 
only reactions with unit jumps, for systems possessing jumps 
of length $2$ or higher, such a direct coupling would cause 
non-physical results. This scenario is analogous to 
\emph{ghost forces} which arise in quasi-continuum methods used 
in the multiscale modelling of materials~\cite{badia2008atomistic}.  
Overlap regions are also necessary for coupling Brownian dynamics 
(SDEs) with mean-field partial differential equations~\cite{FFCE13}.

The paper is organised as follows. In Section \ref{sec:prelims} 
after reviewing the CME/SSA and CLE/CFPE formalisms we introduce 
blending functions and the hybrid jump-diffusion formalism.  
In Section \ref{sec:derivation} we derive weak error bounds for 
the hybrid scheme in the limit of large volume, and in particular 
show that the hybrid scheme does not perform worse than the CLE 
in this regime.  In Section \ref{sec:algorithm} we describe 
three possible discretisations of the process, which can be used in 
practise to simulate the jump-diffusion process. A number 
of numerical experiments which demonstrates the use of the hybrid 
scheme are detailed in Sections  \ref{sec:lotka_volterra},
\ref{sec:dimer_steady_state} and \ref{sec:exit_time}.

\section{Preliminaries}
\label{sec:prelims}

\noindent
Consider a biochemical network of $N$ chemical species interacting 
via $R$ reaction channels within an isothermal reactor of fixed 
volume $V$.   For $i = 1,2, \ldots, N$, denote by $X_i(t)$ the 
number of molecules of species $S_i$ at time $t$, and let 
$\mathbf{X}(t) = (X_1(t),X_{2}(t),\ldots, X_N(t))$.   Under 
the assumption that the chemical species are well-mixed it 
can be shown \cite{gillespie1992rigorous} that $\mathbf{X}(t)$ is 
a continuous time Markov process.  When in state $\mathbf{X}(t)$, 
the  $j$-th  reaction gives rise to a transition 
$\mathbf{X}(t) \rightarrow {\mathbf{X}(t)} + \boldsymbol{\nu}_j$ 
with exponentially distributed waiting time with inhomogeneous 
rate $\lambda_j(\mathbf{X}(t))$, where $\lambda_j(\cdot)$ 
and $\boldsymbol{\nu}_j \in \mathbb{Z}^N$ denote the propensity 
and stoichiometric vector corresponding to the $j$-th reaction,
respectively. More specifically, each reaction is of the 
form
$$
{\mu_r}_1 X_1 + {\mu_r}_2 X_2 + 
\ldots {\mu_r}_N X_N \xrightarrow{k_r} {{\mu_r}_1}' X_1 
+ {{\mu_r}_2}' X_2 + \ldots {{\mu_r}_N}' X_N,
$$
where $r = 1,2, \ldots, R,$ and
$\mu_{ri},$ ${{\mu_r}_1}' \in \mathbb{N}=\{0,1,2,\dots\}$, 
for $i=1,2,\ldots, N$. Let us denote
$\boldsymbol{\mu}_r = (\mu_{r1}, \mu_{r2}, \dots, \mu_{rN})$
and
$\boldsymbol{\mu}_r' = (\mu_{r1}', \mu_{r2}', \dots, \mu_{rN}')$.
The stoichiometric vectors 
$\boldsymbol{\nu}_1,\boldsymbol{\nu}_{2}, \ldots, \boldsymbol{\nu}_R$
are then given by
$$
\boldsymbol{\nu}_r = \boldsymbol{\mu}_r' - \boldsymbol{\mu}_r
$$
and describe the net change in molecular copy numbers 
which occurs during the $r$-th reaction. 
Under the assumption of mass action kinetics, 
the propensity function $\lambda_r$ for the $r$-th reaction is 
$$
\lambda_r(\mathbf{x}) 
= 
k_{r} \prod_{j=1}^{N} \frac{x_j!}{(x_j - \mu_{rj})!},
$$
assuming that $n! = 1$ if $n \leq 0$, to simplify notation. 
Within the interval $[t, t +  \mbd t)$,  
we update
$
\mathbf{X}(t) \rightarrow \mathbf{X}(t) + \boldsymbol{\nu}_j, \mbox{ with probability } \lambda_j(\mathbf{X}(t))\,\mbd t + o(\mbd t).
$
The process $\mathbf{X}(t)$ can thus be expressed as the sum of $R$ 
Poisson processes with inhomogeneous rates $\lambda_j(\mathbf{X}(t))$. 
As noted in \cite{gillespie2000chemical,kurtz1981approximation},
$\mathbf{X}(t)$ can be expressed as a random time change of unit 
rate Poisson processes,
\begin{equation} \label{eq:poisson_rep}
\mathbf{X}(t)=\mathbf{\mathbf{X}(0)}+\sum_{r=1}^{R}{P_{r}}\left(\int_{0}^{t}\lambda_{r}(\mathbf{\mathbf{X}(s)})ds\right)\boldsymbol{\nu}_{r},
\end{equation}
where  ${P_{r}}$ are independent unit-rate Poisson processes. 
This is a continuous time Markov process with infinitesimal 
generator 
\begin{equation}
\mathcal{L}_0f(\mathbf{x}) 
= 
\sum_{r=1}^{R}\lambda_r(\mathbf{x})(f(\mathbf{x} + \boldsymbol{\nu}_r) - f(\mathbf{x})).
\label{notrescaledgenerator}
\end{equation}
The classical  method for sampling realisations of $\mathbf{X}(t)$ 
is the Gillespie SSA \cite{Gillespie:1977:ESS}. Given the 
current state $\mathbf{X}(t)$ at time $t$, the time of next 
reaction $t + \tau$ and state $\mathbf{X}(t + \tau)$ are sampled 
as follows: 
\begin{enumerate}
\item Let $\lambda_0 = \sum_{r=1}^R \lambda_r(\mathbf{X}(t))$.
\item Sample $\tau \sim -\log(u)/\lambda_0$, where $u \sim U[0,1]$.
\item Choose the next reaction $r$ with probability $\lambda_r(\mathbf{X}(t))/\lambda_0$, where $r= 1,2,\ldots, R$.
\item $\mathbf{X}(t + \tau) = \mathbf{X}(t) + \boldsymbol{\nu}_r$.
\end{enumerate}
We note that in advancing the system from time $t$ to time 
$t+\tau$ one needs to generate two random numbers each time.  
Based on the time changed representation (\ref{eq:poisson_rep}) 
one can derive an alternative algorithm, known as the 
\emph{Next Reaction Method} of Gibson and 
Bruck~\cite{Gibson:2000:EES}. Indeed, for a fixed realisation 
of each unit rate Poisson process $P_1, P_2, \ldots, P_R$, define 
$F_r(t)$ to be the last jump time of $P_r$ before time $t$. 
Then for each $r$, the next jump time of $P_r$ after time $F_r(t)$ 
will be distributed as $F_r(t)  - \log u$, where $ u \sim U[0,1]$. 
Clearly, as the process $X(t)$ evolves, the $r$-th reaction will 
then 
occur at $t+\tau_r$ satisfying
$$
T_r(t + \tau_r) = F_r(t) - \log u,
\qquad\quad
\mbox{where}
\qquad\quad
T_r(t) = \int_0^t \lambda_r(\mathbf{X}(s))\,\mbox{d}s.
$$
This provides the basis of the Next Reaction method.  
Suppose we are time $t$, the next reaction will occur at 
$t + \tau_{min}$ for
$$
\tau_{min}  
= 
\mathop{\mbox{argmin}}_{r \in \lbrace 1,2, \ldots, R\rbrace}
\left\lbrace \tau_r \, :  F_r(t) - \log u_r 
= T_r(t + \tau_r) \right\rbrace,
$$
where $u_r \sim U[0,1]$ are independently distributed random 
numbers. Noting that the values of the propensities do not 
change within $[t, t + \tau_{min})$ we have 
$T_r(t + \tau) = T_r(t) + \tau_r \lambda_r(\mathbf{X}(t))$, 
so that the next reaction time is given by
$$
\tau_{min}  
= 
\mathop{\mbox{argmin}}_{r\in\lbrace1,2,\ldots, R\rbrace}
\left\lbrace\, \frac{F_r(t) - \log u_r - T_r(t)}{\lambda_r(\mathbf{X}(t))}
\right\rbrace,
$$
at which time the reaction that occurs is the one for which 
$\tau_r = \tau_{min}$. This leads to the following exact 
algorithm for sampling realizations of $\mathbf{X}(t)$.
\begin{enumerate}
\item Set the initial number of molecules of each species, 
set $t=0$.
\item Calculate the propensity function $\lambda_{r}$ for each reaction.
\item Generate $R$ independent random numbers $u_{r}\sim U[0,1]$.
\item Set $F_{r}=-\log(u_{r})$ and $T_r = 0$
      for each $r=1,2,\dots,R$.
\item Set $\tau_{r}=(F_{r}-T_{r})/\lambda_{r}$
      for each $r=1,2,\dots,R$.
\item Set $\tau_{min}=\min_{r}\{\tau_{r}\}$ and let ${\mu}$ be the reaction for which this minimum is realised.
\item Set $t=t+\tau_{min}$ and update the number of each molecular species according to reaction $\mu$.
\item For each k, set $T_{r}=T_{r}+\lambda_{r}\tau_{min}$, and for the reaction $\mu$, let $u \sim U[0,1]$ and set $F_{r}=F_{r}-\log{u}$.
\item Recalculate the propensity functions $\lambda_{r}$.
\item Return to step 5 or quit.
\end{enumerate}
This algorithm was introduced by  
Gibson and Bruck~\cite{Gibson:2000:EES} who additionally proposed 
the introduction of an indexed priority queue to efficiently search 
for the minimum required in step 6, along with a dependency graph 
structure to efficiently update propensity values in step 9.  
This makes it less computationally intensive from the Gillespie SSA
when simulating systems with many reaction channels~\cite{Cao:2004:EFS}.

\subsection{Diffusion Approximation}

\noindent
For $r=1,2, \ldots, R$ define $\widetilde{\lambda}_r(\mathbf{x})$ 
to be a smooth, non-negative extension of $\lambda_r(\mathbf{x})$ 
from $\mathbb{N}^N$ to $\mathbb{R}^N$ (the precise conditions on 
this extension are given in Section \ref{sec:derivation}).  
Given the extended propensities, a commonly used approximation 
of (\ref{eq:poisson_rep}) is the CLE, given by the 
following It\^{o} SDE
\begin{equation}
\label{eq:cle}
\mbd \mathbf{Y}(t) 
= 
\sum_{r=1}^{R}
\boldsymbol{\nu}_r \,
\widetilde{\lambda}_r(\mathbf{Y}(t))\,\mbd t 
+ \sum_{r=1}^{R}\boldsymbol{\nu}_r \,
\sqrt{\widetilde{\lambda}_r(\mathbf{Y}(t))}\, \mbd W_r(t),
\end{equation}
where $W_r(t)$ are mutually independent standard Brownian motions. 
This diffusion approximation is valid in the large volume regime, 
where all species exist in abundance, and all reactions occur 
frequently, see for example~\cite{gillespie2000chemical,MBZ10}.  
More precisely, one can show strong convergence of $\mathbf{Y}(t)$ 
to $\mathbf{X}(t)$ over finite time intervals $[0, T]$, 
see~\cite{anderson2011continuous,ganguly2014jump}.  
Clearly, the lifting of $\lbrace \lambda_r \rbrace_{r=1}^R$ 
to $\lbrace \widetilde{\lambda}_r \rbrace_{r=1}^R$ is not unique,  
and different extensions will give rise to different diffusion 
approximations.  However, as we shall see in 
Section \ref{sec:derivation}, in the classical large 
volume rescaling,  the dynamics of the process will be 
largely determined by the value of the propensities on 
the rescaled grid $\mathbb{e}\mathbb{N}^N$, and indeed, 
subject to the extension satisfying a number of assumptions, 
different extensions will lead to the diffusion approximation 
having  weak error of the same order.

\subsection{The Hybrid Scheme}
\label{subsec:stoc_fram}

\noindent
In this section, we introduce a jump-diffusion process which 
provides an approximation which is intermediate between the 
Gillespie SSA and CLE by introducing a series of \emph{blending} 
functions $\beta_1, \beta_{2}, \ldots, \beta_R$ which are used 
to blend the dynamics linearly between the SSA jump process and the 
CLE. More specifically, given $R$ smooth functions
$\beta_r:\mathbb{R}^d\rightarrow [0,1]$, $r=1,2, \ldots, R$ 
we consider the following It\^o jump-diffusion equation
\begin{equation}
\label{eq:hybrid}
\begin{aligned}
\mathbf{Z}(t) 
= 
\mathbf{Z}(0)
&+ \sum_{r=1}^{R}
P_{r}\left(\int_{0}^{t}
\beta_r(\mathbf{Z}(s))
\, \lambda_r(\llbracket \mathbf{Z}(s) \rrbracket) 
\, \mbd s\right)\boldsymbol{\nu}_{r}
\\
&+ 
\sum_{r=1}^{R} \boldsymbol{\nu}_r 
\int_{0}^{t} \Big(
1-\beta_r(\mathbf{Z}(s)) \Big) \,
\widetilde{\lambda}_r(\mathbf{Z}(s))\,\mbd s \\
&+
\sum_{r=1}^{R} \boldsymbol{\nu}_{r}
\int_{0}^{t}
\sqrt{ \Big(
1 - \beta_r(\mathbf{Z}(s)) \Big)
\, \widetilde{\lambda}_r(\mathbf{Z}(s))}
\,\mbd W_{r}(s),
\end{aligned}
\end{equation}
where $\lbrace W_r \rbrace_{r=1}^{R}$ and 
$\lbrace P_r\rbrace_{r=1}^R$ 
are standard Wiener and Poisson processes, respectively, all 
mutually independent, and $\llbracket \mathbf{x} \rrbracket$ 
is the closest point 
in the lattice $\mathbb{Z}^N$ to $\mathbf{x} \in \mathbb{R}^N$. 
Thus, (\ref{eq:hybrid}) describes a jump-diffusion Markov process 
with infinitesimal generator $\mathcal{G}$ defined by
\begin{eqnarray}
\mathcal{G}f(\mathbf{z}) \!\!\!&=&\!\!\!
\sum_{r=1}^{R} \!\beta_r(\mathbf{z})
\lambda_r(
\llbracket \mathbf{z} \rrbracket)
\left[f(\mathbf{z} + \boldsymbol{\nu}_r) - f(\mathbf{z})\right]
+ 
\sum_{r=1}^{R}(1 - \beta_r(\mathbf{z})) 
\, \widetilde{\lambda}_r(\mathbf{z}) \, \boldsymbol{\nu}_r
\!\cdot\!\nabla f(\mathbf{z}) \nonumber \\
\!\!\!&+&\!\!\! \frac{1}{2} \sum_{r=1}^{R}
(1-\beta_r(\mathbf{z}))
\,\widetilde{\lambda}_r(\mathbf{z})
\,(\boldsymbol{\nu}_r\otimes  \boldsymbol{\nu}_r) 
: \nabla\nabla f(\mathbf{z}),
\label{eq:hybrid_generator}
\end{eqnarray}
for all $f \in C^2_0(\mathbb{R}^N)$,
where $\otimes$ stands for the tensor product, 
and $A:B=\text{trace}(A^{T}B)$ for square matrices $A$ and $B$. 
The first term on the right hand side of (\ref{eq:hybrid_generator}) 
captures the jump behaviour of the process, while the remaining two
terms encode the effect of the diffusive dynamics.  

We see that in regions where $\beta_r(x) = 1$ the dynamics of 
the $r$-th reaction is modelled by a pure jump process. Conversely, 
when $\beta_r(x) = 0$ the dynamics are purely diffusive, corresponding 
to CLE dynamics. In intermediate regions where $0 < \beta_r(x) < 1$ 
we obtain a mixture of the two. The rationale is to choose 
$\beta_r$ to be $0$ in regions where the CLE provides a valid 
approximation of the biochemical system, and $\beta_r$ to be $1$ in 
regions where the diffusion approximation breaks down, i.e. in regions 
where the concentrations of certain species are low and the discrete 
behaviour becomes significant. An example of a trajectory is shown 
in Figure \ref{fig:vanity}, where the blue line depicts diffusive 
dynamics, while the red lines indicate jumps.  We note that the 
process $\mathbf{Z}(t)$ can still attain negative 
(and thus non-physical) states, however, 
$\llbracket \mathbf{Z}(t) \rrbracket$ is always non-negative.  
A natural interpretation is that one should consider the 
cell $\llbracket \mathbf{Z}(t) \rrbracket$ as the actual 
observed dynamics of the system, and the state of the underlying 
process $\mathbf{Z}(t)$ is a hidden Markov model which is not 
observed directly.

\begin{figure} [t]
\centering
\includegraphics[width=\textwidth]{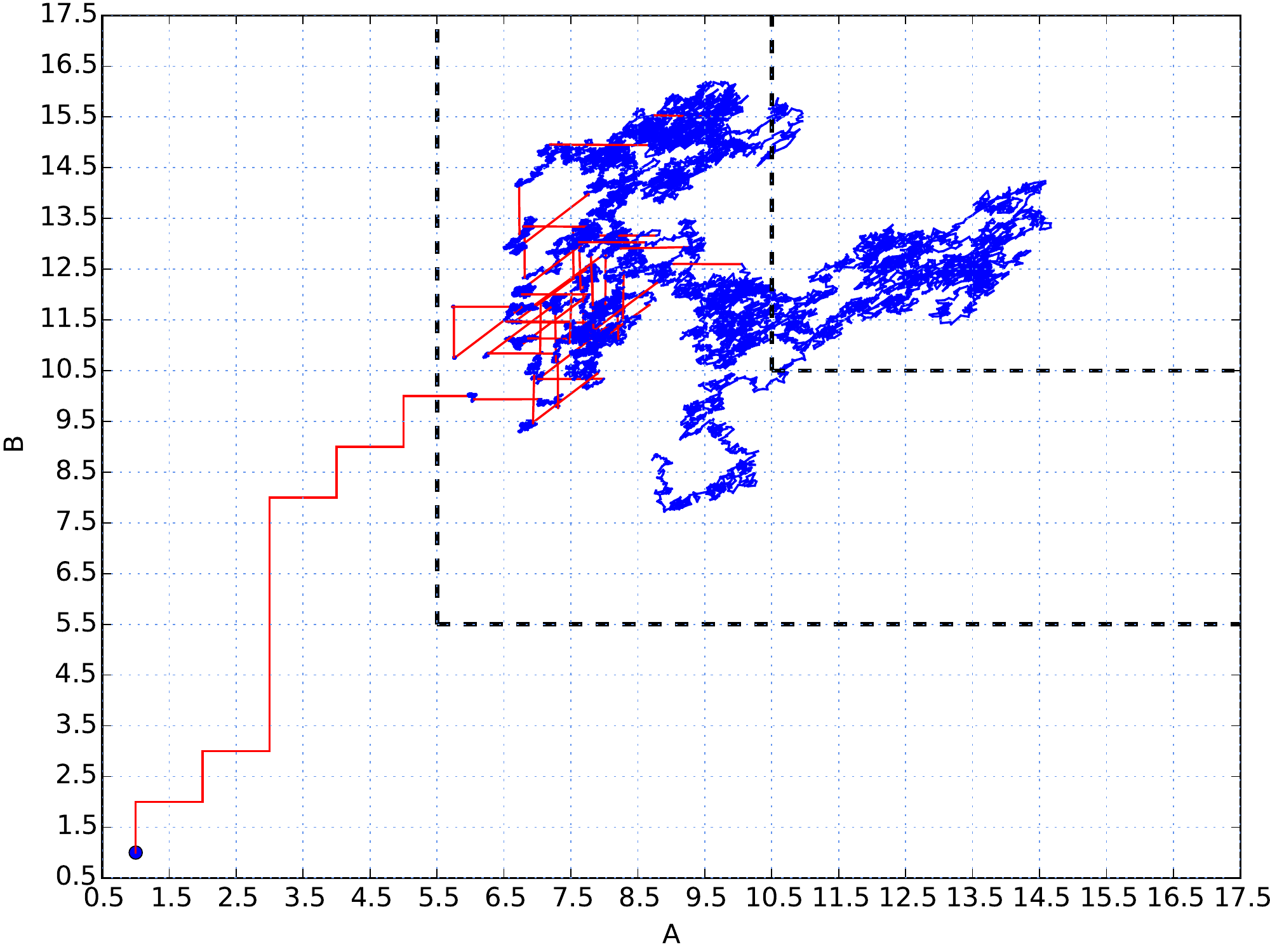}
\caption{\emph{A typical trajectory for hybrid approximation of 
the two-species chemical system described in $(\ref{eq:dimer})$
starting from $(A, B) = (1,1)$.  Red lines denote jump dynamics 
while blue lines CLE dynamics. The black lines demarcate 
the ``blending'' region.}}
\label{fig:vanity}
\end{figure}

\subsection{Choice of Blending Functions}
\label{sec:blending}

\noindent
The blending functions $\beta_1(x), \beta_2(x), \ldots, \beta_R(x)$ 
are to be considered as simulation parameters which are chosen to 
balance the trade-off between the computational cost of using the
SSA and the error arising from the diffusion approximation. 
Generally chosen so that Langevin dynamics are only used in regions 
where the reactions are considered fast. Since not all species 
are involved in every reaction, it is natural to choose each 
blending function differently.  For a single species system, 
a natural choice of blending function is the following piecewise 
linear function
$$
\beta(x, I_1, I_2) 
= 
\begin{cases} 1, \rule{0pt}{2mm} & \mbox{ if } x \le I_1,  \\  
\displaystyle \frac{I_2-x}{I_2-I_1}, \rule{0pt}{6.3mm} & \mbox{ if } I_1 < x < I_2, \\ 
0, \rule{0pt}{5mm} & \mbox{ if }x \ge I_2,\end{cases}
$$
where $0 < I_1 < I_2$ are the boundaries between the different 
regions.  With this choice of $\beta(\cdot, I_1, I_2)$, the 
hybrid process will perform purely jump dynamics for 
$0 < \mathbf{Z}(t) < I_1$, purely diffusive dynamics 
for $I_2 < \mathbf{Z}(t) < \infty$, and jump-diffusion 
in between.

For chemical systems with $N$ species we can construct blending 
functions for each reaction as follows.  Let $S_r$ be the set 
of chemical species involved in the $r$-th reaction (both as 
reactants, and products of the reaction).  Then we can 
define $\beta_r(\mathbf{x})$, $r=1,2,\dots,R$, as follows
\begin{equation}
\label{eq:multi_species_blending}
\beta_r(\mathbf{x}) = 1 - \prod_{n \in S_r}(1 - \beta(x_n, I^n_1, I^n_2)),
\end{equation}
where $I^n_1 < I^n_2$, $n=1,2,\ldots, N$, are the boundaries 
for each individual chemical species.  With this choice of blending 
function, when $\mathbf{Z}(t)$ is in a state where one of the species 
involved in the $r$-th reaction is not abundant, the hybrid process 
will blend between jumps and diffusion to simulate the corresponding reaction.
The choice of the boundaries $I_1^n$ and $I_2^n$ is important 
to correctly delineate between discrete and continuous behaviour. 
As will be seen in the numerical experiments in 
Section~\ref{sec:dimer_steady_state}, the accuracy of the 
scheme is dependent on the width of the blending region.  If the 
blending region $I_2^n - I_1^n$ is small (i.e. close to 1), 
the step size for the discretised CLE must be accordingly decreased 
to maintain constant error.  

\begin{remark}
While $(\ref{eq:multi_species_blending})$ is a natural blending function 
for the typical biochemical systems where the diffusion approximation 
breaks down only close to the boundary of the positive orthant, for 
constrained systems, one must have jump dynamics near all the boundaries 
of the system.  A typical example is the reversible isomerisation model 
$$
{\mbox{ \raise 0.851 mm \hbox{$A$}}}
\;
\mathop{\stackrel{\displaystyle\longrightarrow}\longleftarrow}^{k_1}_{k_2}
\;
{\mbox{\raise 0.851 mm\hbox{$B$,}}}
$$ 
which can be modelled as a single species 
birth-death process for $A$ with birth rate $\lambda(m) = k_2(M - m)$ and 
death rate $\mu(m) = k_1 m$, where $M$ is the total molecule count.  
For such systems, an appropriate blending function would take value 
$1$ in neighbourhoods of both $m = 0$ and $m = M$.
\end{remark}

\section{Derivation and Consistency of the Hybrid Scheme}
\label{sec:derivation}

\noindent
In this section we shall make explicit the regime in which the hybrid 
scheme correctly captures the dynamics of the original process, and 
subsequently derive weak error estimates for {the} expectation of 
observables at a finite time $T > 0$.  Let $\mathbb{L} := \mathbb{N}^N$ 
be the lattice of possible states.  For most biochemical systems, we 
can make the following natural assumptions.
\begin{assumption}
The propensity functions satisfy $\lambda_j(x) \geq 0$ for 
all $x \in \mathbb{L}$, and $\lambda_j(x) = 0$ if $x \in \mathbb{L}$ 
and $x + \boldsymbol{\nu}_j \in \mathbb{Z}^N \setminus \mathbb{L}$.
\end{assumption}

\noindent
In particular, we are ensured that the jump process $\mathbf{X}(t)$ 
never attains a negative state. 

\noindent
\begin{assumption}
For fixed $T > 0$, the set $\lbrace \mathbf{X}(t) \,|\, t \in [0,T]\rbrace$ 
is almost surely constrained within a bounded subset $\Omega$ of $\mathbb{L}$. 
\end{assumption}
\noindent
We note that the domain $\Omega$ will depend on the initial condition 
$\mathbf{X}(0) = \mathbf{X}_0$. As it stands this assumption will not 
hold for general chemical reacting systems. Under suitable conditions 
on the propensities it is possible to replace this assumption with 
a localization result showing that the probability of $\mathbf{X}(t)$ 
escaping the bounded set is exponentially small. We shall avoid this 
approach for simplicity, simply noting that one can always ensure 
this assumption by setting propensities to zero outside a fixed 
bounded region. 

Chemical Langevin dynamics $\mathbf{Y}(t)$ are only a valid approximation 
of $\mathbf{X}(t)$ in the large volume limit.  To study this regime, 
we introduce a system size $V \gg 1$ which can be viewed as the 
(dimensionless) volume of the reactor. Writing $\varepsilon = V^{-1}$, 
we then rewrite the molecular copy number $\mathbf{X}(t)$ as 
$\varepsilon^{-1} \mathbf{X}^\varepsilon(t)$ where 
$\mathbf{X}^{\varepsilon}(t)$ will be the vector of concentrations 
of each chemical species. We shall assume that each rate constant 
$k_r$ satisfies
$$
k_r = d_r \, \varepsilon^{-z_r},	
\quad \mbox{ where } d_r > 0
\quad 
\mbox{ and } 
\quad
z_r =-1+ \sum_{k=1}^K \mu_{rk}.
$$ 
Given this scaling assumption, we can always write the 
propensity for the $r$-th reaction, $r=1,2,\dots,R$, as 
$$
\lambda_r
\left(
\frac{\mathbf{x}}{\varepsilon}
\right) 
= 
\frac{1}{\varepsilon}
\, \lambda^\varepsilon_r(\mathbf{x}), 
\quad \mbox{ for } \quad
\mathbf{x} \in \varepsilon \mathbb{L}.
$$
where $\lambda^\varepsilon_r(\mathbf{x})$ is $O(1)$ with 
respect to $\varepsilon$. Using (\ref{notrescaledgenerator}),
the generator of the rescaled process $\mathbf{X}^\varepsilon(t)$ 
is given as follows:
\begin{align}
\label{eq:generator_epsilon}
\mathcal{L}^\varepsilon f(\mathbf{x}) 
&= \frac{1}{\varepsilon}
\sum_{r=1}^{R}\lambda^\varepsilon_r(\mathbf{x})(f(\mathbf{x} + \varepsilon 
\boldsymbol{\nu}_r) - f(\mathbf{x})),	
\quad \mbox{ for } \quad \mathbf{x} \in \varepsilon\,{\mathbb{L}}.
\end{align}
We now introduce the hybrid jump-diffusion scheme. To do so we must 
extend propensities $\lambda^\varepsilon(\mathbf{x})$ from the discrete 
lattice $\varepsilon\mathbb{L}$ to 
$\widetilde{\lambda}^\varepsilon(\mathbf{x})$ on the continuous 
space $\mathbb{R}^N$. We shall make the following assumptions 
on the extension.
\begin{assumption}
\label{ass:extension_assumptions}
The following properties hold for the extended 
propensities $\widetilde{\lambda}_j^\varepsilon(x)$:
\begin{enumerate}
\item They are non-negative, and lie in $C^3(\mathbb{R}^N)$. 
\item They are bounded, uniformly with respect to $\varepsilon$, 
and the same applies for their mixed derivatives up to order $3$.
\item For each $j$, $\widetilde{\lambda}_j^\varepsilon(x)$ is 
zero outside a bounded domain $\Omega^\varepsilon $ which 
contains $\varepsilon \, \Omega$.
\end{enumerate}
\end{assumption}
\begin{remark}
A $C^3(\mathbb{R}^N)$ extension of the propensities satisfying 
Assumption $\ref{ass:extension_assumptions}$ is always possible. 
Indeed, for each $j$, set $\widetilde{\lambda}_j$ to be zero in 
$\mathbb{R}^N\setminus \Omega^\varepsilon$.  Then one can extend 
the value of the propensities to $\Omega^\varepsilon$ by transfinite 
interpolation, see~\cite{gordon1971blending, barnhill1984smooth}.
\end{remark}

\begin{remark}
Such an extension may result in propensities which differ from 
the ``standard'' propensities typically used for the CLE. 
In particular, propensities of the form $k_1 x(x-1)$ must be modified 
so as to remain non-negative. Such an explicit construction 
of extended propensities for unimolecular and bimolecular reactions 
of a single species can be found 
in~\cite[Example 4.7-4.8]{karlsson2011towards}.
\end{remark}

\noindent
Using the extended propensities, one can  extend the Markov jump 
process $\mathbf{X}^\varepsilon(t)$ to take initial conditions 
$\mathbf{X}^\varepsilon(0) \in \Omega^\varepsilon$.  The 
infinitesimal generator of $\mathbf{X}^\varepsilon(t)$ is 
the natural extension of (\ref{eq:generator_epsilon}), 
also denoted by $\mathcal{L}^\varepsilon$ defined by
$$
{\mathcal{L}}^\varepsilon f(\mathbf{x}) 
= 
\frac{1}{\varepsilon}\sum_{r=1}^{R} 
\widetilde{\lambda}^\varepsilon_r(\mathbf{x})
\left[f(\mathbf{x} + \varepsilon \boldsymbol{\nu}_r) - f(\mathbf{x})\right],	
\quad \text{for all} \ f \in C_0(\mathbb{R}^N).
$$
For a fixed observable $g \in C^3(\mathbb{R}^N)$, 
define the value function 
$u^\varepsilon:[0,T]\times\mathbb{R}^N \rightarrow \mathbb{R}$:
$$
u^\varepsilon(t, \mathbf{x}) 
= 
\mathbb{E}\left[g(\mathbf{X}^\varepsilon(t)) \, | 
\, \mathbf{X}^\varepsilon(0) = \mathbf{x}\right].
$$
Then $u^\varepsilon(t,\mathbf{x})$ can be expressed as 
the unique solution of the Backward Kolmorogov 
equation~\cite{Gardiner:1985:HSM}
\begin{equation}
\label{eq:rescaled_kbe_u}
\begin{aligned}
\partial_t u^\varepsilon(t, \mathbf{x}) 
&= 
{\mathcal{L}}^\varepsilon u^\varepsilon(t,\mathbf{x}),	
\quad \mbox{ for } (t, \mathbf{x}) \in (0, T)\times \mathbb{R}^N,
\\ 
u^\varepsilon(0, \mathbf{x}) 
&= g(\mathbf{x}),
\quad \mbox{ for }  \mathbf{x} \in \mathbb{R}^N.
\end{aligned}
\end{equation}
For any fixed $\mathbf{x} \in \mathbb{R}^N$, equation
(\ref{eq:rescaled_kbe_u}) can be viewed as an infinite 
linear ODE on the translated lattice 
$\mathbf{x} + \varepsilon \mathbb{\mathbb{L}}$.  
By Assumption \ref{ass:extension_assumptions} the propensity 
is only non-zero for finitely many terms, thus the dynamics is 
characterised by a finite linear system of ODEs.  Existence and 
uniqueness of a solution $u^\varepsilon(t,\mathbf{x})$ in $C^1[0,T]$ 
follows immediately.  Moreover, if $g(\textbf{x})$ is locally bounded, 
then so is $u^\varepsilon(t,\mathbf{x})$. Clearly, for $\mathbf{x} \in \mathbb{R}^N$ such that 
$\widetilde{\lambda}^\varepsilon_r(\mathbf{x}) = 0$ 
for all $r$ we have $u^\varepsilon(t, \mathbf{x}) = g(\mathbf{x})$, 
for all $t \in [0,T]$.

Moreover, using Assumption~\ref{ass:extension_assumptions}, 
it follows that  $u^\varepsilon(t, \cdot) \in C^3(\mathbb{R}^N)$, 
such that the mixed derivatives can be expressed as the unique 
solutions of the following equations, where $\partial_{i}$, 
$\partial_{ij}$ and $\partial_{ijk}$ denote first, second 
and third spatial derivatives with respect to the variables 
$x_i, x_j, x_k$, 
$\forall i, j, k \in \lbrace 1, 2,\ldots, N\rbrace$.
\begin{equation}
\label{eq:D1}
\partial_t  \partial_i u^\varepsilon(t,\mathbf{x}) 
- {\mathcal{L}}^\varepsilon \partial_{i} u^\varepsilon(t, \mathbf{x}) 
= 
\frac{1}{\varepsilon}\sum_{r=1}^{R} \partial_{i} \widetilde{\lambda}^\varepsilon_r(\mathbf{x}) \left[u^\varepsilon(t, \mathbf{x} + \varepsilon \boldsymbol{\nu}_r) - u^\varepsilon(t, \mathbf{x})\right],
\end{equation}
\begin{align}
\notag\partial_t \partial_{ij} u^\varepsilon(t,\mathbf{x}) 
&- {\mathcal{L}}^\varepsilon \partial_{ij} 
u^\varepsilon(t, \mathbf{x}) 
= \frac{1}{\varepsilon}\sum_{r=1}^{R} \partial_{ij} \widetilde{\lambda}^\varepsilon_r(\mathbf{x}) 
\left[u^\varepsilon(t, \mathbf{x} + \varepsilon \boldsymbol{\nu}_r) 
- u^\varepsilon(t, \mathbf{x})\right] \\
\notag		
&\quad + \frac{1}{\varepsilon}\sum_{a,b\in\Pi_2}\sum_{r=1}^{R} 
\partial_a \widetilde{\lambda}^\varepsilon_r(\mathbf{x}) 
\left[\partial_b u^\varepsilon(t, \mathbf{x} 
+ \varepsilon \boldsymbol{\nu}_r) - \partial_b u^\varepsilon(t, \mathbf{x})\right],
\end{align}
\begin{align}
\notag\partial_t \partial_{ijk} u^\varepsilon(t,\mathbf{x}) 
&- {\mathcal{L}}^\varepsilon \partial_{ijk} u^\varepsilon(t, \mathbf{x}) 
= 
\frac{1}{\varepsilon}
\sum_{r=1}^{R} \partial_{ijk} 
\widetilde{\lambda}^\varepsilon_r(\mathbf{x}) 
\left[u^\varepsilon(t, \mathbf{x} + \varepsilon \boldsymbol{\nu}_r) 
-  u^\varepsilon(t, \mathbf{x})\right]
\\
\notag &\quad + 
\frac{1}{2\varepsilon}
\sum_{a,b,c\in\Pi_3}\sum_{r=1}^{R} \partial_{ab} \widetilde{\lambda}^\varepsilon_r(\mathbf{x}) 
\left[
\partial_c u^\varepsilon(t, \mathbf{x} 
+ \varepsilon \boldsymbol{\nu}_r) 
- \partial_c u^\varepsilon(t, \mathbf{x})\right]
\\
\notag &\quad + \frac{1}{2\varepsilon}
\sum_{a,b,c\in\Pi_3}\sum_{r=1}^{R} \partial_a \widetilde{\lambda}^\varepsilon_r(\mathbf{x}) 
\left[\partial_{bc} u^\varepsilon(t, \mathbf{x} 
+ \varepsilon \boldsymbol{\nu}_r) - \partial_{bc} 
u^\varepsilon(t, \mathbf{x})\right],
\end{align}
where $\Pi_2$ and $\Pi_3$ denote the set of 
permutations of $\lbrace i, j \rbrace$ and 
$\lbrace i, j, k \rbrace$, respectively.

\begin{lemma}
Given $p, q, r = 1, 2,\ldots, R$, and $t \in [0,T]$ define 
the following scalar quantities
\begin{equation*}
\begin{aligned}
A_r^\varepsilon(t, \mathbf{x}) 
&= \boldsymbol{\nu}_r\cdot\nabla u^\varepsilon(t, \mathbf{x}) 
= \sum_{i=1}^{N} \nu_{r, i} 
\partial_{i} \, u^\varepsilon(t, \mathbf{x}), 
\\
B_{p, q}^\varepsilon(t, \mathbf{x}) 
&= (\boldsymbol{\nu}_p\otimes\boldsymbol{\nu}_q) 
:
\nabla \nabla u^\varepsilon(t, \mathbf{x}) 
= \sum_{i,j=1}^{N} \nu_{p, i} \, \nu_{q,j} \, 
\partial_{ij} u^\varepsilon(t, \mathbf{x}),
\\
C_{p, q, r}^\varepsilon(t, \mathbf{x}) 
&= 
(\boldsymbol{\nu}_p\otimes\boldsymbol{\nu}_q\otimes\boldsymbol{\nu}_r) 
: 
\nabla \nabla\nabla u^\varepsilon(t, \mathbf{x}) 
= 
\sum_{i,j,k=1}^{N} 
\nu_{p, i} \, \nu_{q,j} \, \nu_{r,k} \, \partial_{ijk} 
u^\varepsilon(t, \mathbf{x}).
\end{aligned}
\end{equation*}
Then, there exists constants $K_1$, $K_2$ and $K_3$ and 
$C_1$, $C_2$, and $C_3$ independent of $\varepsilon$ 
such that
\begin{equation}
\label{eq:uniform_bound}
\lVert{ A_r^\varepsilon(t, \cdot)}\rVert_{\infty} \leq C_1 e^{K_1 T}\!,
\quad
\lVert{ B_{p, q}^\varepsilon(t, \cdot)}\rVert_{\infty} \leq C_2 e^{K_2 T}\!,
\quad
\lVert{ C_{p, q, r}^\varepsilon(t, \cdot)}\rVert_{\infty} \leq C_3 e^{K_3 T}\!,
\end{equation}
for $t \in [0,T]$, where the $K_1$, $K_2$ and $K_3$ 
depend on  
$$
\left\lVert\widetilde{\lambda}^\varepsilon_r\right\rVert_{C^1(\mathbb{R}^N)},
\quad
\left\lVert\widetilde{\lambda}^\varepsilon_r\right\rVert_{C^2(\mathbb{R}^N)}
\quad \mbox{and} \quad
\left\lVert\widetilde{\lambda}^\varepsilon_r\right\rVert_{C^3(\mathbb{R}^N)},
$$
respectively, which by Assumption $(\ref{ass:extension_assumptions})$ 
are bounded independently of $\varepsilon$. 
\end{lemma}
\begin{proof}
Using (\ref{eq:D1}), we have
$$
\partial_t A_q^\varepsilon(t, \mathbf{x}) - {\mathcal{L}}^\varepsilon A_q^\varepsilon(t,\mathbf{x}) = F_1^\varepsilon(t, \mathbf{x})
$$
where
\begin{equation}
F_1^\varepsilon(t, \mathbf{x}) =
\frac{1}{\varepsilon} \sum_{r=1}^{R} \boldsymbol{\nu}_q
\cdot
\nabla\widetilde{\lambda}_r^\varepsilon(\mathbf{x}) 
\left[u^\varepsilon(t, \mathbf{x} + \varepsilon \boldsymbol{\nu}_r) 
- u^\varepsilon(t, \mathbf{x})\right].
\label{F1def}
\end{equation}
Let $P^\varepsilon_t$ be the semigroup associated with
$\mathbf{X}^\varepsilon(t),$ so that 
$$
P^\varepsilon_tf(\mathbf{x}) 
= 
\mathbb{E}\left[
f(\mathbf{X}^\varepsilon(t)) \, | {X}^\varepsilon(0) = \mathbf{x}
\right].
$$  Clearly, if $f$ is bounded, then 
$|P^\varepsilon_t f(\mathbf{x}) | \leq \lVert f \rVert_{\infty}.$ 
In particular, $|u^\varepsilon(t, \mathbf{x})|$ is bounded uniformly 
with respect to $\varepsilon$. It is straightforward to check 
that we can write the solution $A_q^\varepsilon(t, \mathbf{x})$ 
as 
$$
A_q^\varepsilon(t, \mathbf{x}) 
= 
P_t 
\left[
\boldsymbol{\nu}_q\cdot\nabla g\right]
(\mathbf{x}) 
+ 
\int_0^t P_{t-s} \, F_1^\varepsilon(s, \mathbf{x})\,\mbd s,
$$
where $g$ is the fixed observable used as the initial condition in
(\ref{eq:rescaled_kbe_u}). Thus,
\begin{equation}
|A_q^\varepsilon(t, \mathbf{x})| 
\leq 
\lVert\boldsymbol{\nu}_q\cdot\nabla g\rVert_{\infty} 
+ \int_0^t \lVert F_1^\varepsilon(s, \cdot) \rVert_{\infty}\,\mbd s.
\label{Aqeest}
\end{equation}
Since $u^\varepsilon(t, \cdot)$ is $C^1$, we have 
$$
\frac{1}{\varepsilon}\left[
u^{\varepsilon}(t, \mathbf{x} + \varepsilon\boldsymbol{\nu}_r) 
- u^{\varepsilon}(t, \mathbf{x})\right] 
=
\int_0^1 A_r^\varepsilon(t, \mathbf{x} 
+ w \, \varepsilon \, \boldsymbol{\nu}_r)
w\,\mbox{d}w.
$$
Substituting into (\ref{F1def}), we obtain
$$
\lVert F_1^\varepsilon(t, \cdot)\rVert_{\infty} 
\leq 
C\sum_{r=1}^R\lVert \widetilde{\lambda}^\varepsilon_r\rVert_{C^1} 
\lVert A_r^\varepsilon(t, \cdot)\rVert_{\infty}\,\mbd s.
$$
Therefore, using (\ref{Aqeest}), we get
$$
\max_{q}\lVert A_q^\varepsilon(t,\cdot)\rVert_{\infty} \leq \max_{q}\lVert A_q^\varepsilon(0,\cdot)\rVert_{\infty} + K_1 \int_0^T \max_{q}\lVert A_q^\varepsilon(s,\cdot)\rVert_{\infty}\,\mbd s,
$$
so that
$$\max_{q}\lVert A_q^\varepsilon(t,\cdot)\rVert_{\infty} \leq C_1e^{K_1 T}.$$
Similarly,
\begin{align*}
\partial_t B_{p,q}^\varepsilon(t, \mathbf{x}) 
- {\mathcal{L}}^\varepsilon B_{p,q}^\varepsilon(t,\mathbf{x}) 
&= F^\varepsilon_2(t, \mathbf{x}),
\end{align*}
where
\begin{align*}
F^\varepsilon_2(t, \mathbf{x}) 
&= \frac{1}{\varepsilon} 
\sum_{r=1}^{R} 
(\boldsymbol{\nu}_p \otimes \boldsymbol{\nu}_q) :
\nabla\nabla\widetilde{\lambda}_r^\varepsilon(\mathbf{x}) 
\left[u^\varepsilon(t, \mathbf{x} + \varepsilon \boldsymbol{\nu}_r) 
- u^\varepsilon(t, \mathbf{x})\right]
\\
& + \frac{1}{\varepsilon} 
\sum_{r=1}^{R} \boldsymbol{\nu}_q
\cdot \nabla\widetilde{\lambda}_r^\varepsilon(\mathbf{x})
\left[
A_p^\varepsilon(t, \mathbf{x} + \varepsilon \boldsymbol{\nu}_r) 
- A_p^\varepsilon(t, \mathbf{x})
\right]
\\
& 
+ \frac{1}{\varepsilon} \sum_{r=1}^{R} 
\boldsymbol{\nu}_p
\cdot 
\nabla\widetilde{\lambda}_r^\varepsilon(\mathbf{x}) 
\left[
A_q^\varepsilon(t, \mathbf{x} + \varepsilon \boldsymbol{\nu}_r) 
- A_q^\varepsilon(t, \mathbf{x})\right].
\end{align*}
Thus it follows that
$$
B_{p,q}^\varepsilon(t, \mathbf{x}) 
= 
P_t 
\left[
(\boldsymbol{\nu}_p \otimes \boldsymbol{\nu}_q)
:\nabla\nabla g\right](\mathbf{x}) 
+ \int_0^t P_{t-s} \, F_2^\varepsilon(s, \mathbf{x})\,\mbd s,
$$
where 
$$
\lVert F_2^\varepsilon(t, \cdot)\rVert_{\infty} 
\leq 
C\sum_{r=1}^R
\lVert \widetilde{\lambda}^\varepsilon_r\rVert_{C^2(\mathbb{R}^N)}
\left(||A_r(t, \cdot)||_{\infty}  + ||B_{r,q}(t, \cdot)||_{\infty} 
+ ||B_{r,p}(t, \cdot)||_{\infty} \right).
$$
It follows that
\begin{align*}
\max_{p, q}\left\lVert B_{p, q}^\varepsilon(t, \cdot)\right\rVert_{\infty} 
\leq 
&\max_{p, q}\left\lVert B_{p, q}^\varepsilon(0, \cdot)\right\rVert_{\infty}  
\\ 
&+ K_2\int_0^T 
\left(
\max_{p, q}\left\lVert B_{p, q}^\varepsilon(s, \cdot)\right\rVert_{\infty} 
+ \max_{p}\left\lVert A_p(s, \cdot)\right\rVert_{\infty} \right) \,\mbd s,
\end{align*}
which implies the second inequality in (\ref{eq:uniform_bound}) by 
Gronwall's inequality. The proof of the third inequality in
(\ref{eq:uniform_bound}) follows analogously.
\end{proof}

\subsection{The hybrid scheme in the large-volume limit}

\noindent
Having extended the propensity function to take arbitrary 
values in $\mathbb{R}^N$ we can now study the weak error 
that arises from taking the hybrid approximation in the 
large-volume limit. 

\begin{assumption} 
\label{asblend}
We assume that the blending 
functions $\beta_r$, $r=1,2,\ldots, R$ 
satisfy $\beta_r \in C^0(\mathbb{R}^N)$ and
$\beta_r(\mathbf{x}) \in [0,1]$,
for all $\mathbf{x} \in \mathbb{R}^N$.
\end{assumption}

\noindent
Given the extended propensities, we can apply the same 
large-volume rescaling to the hybrid process 
(\ref{eq:hybrid}) to obtain a jump-diffusion 
$\mathbf{Z}^\varepsilon(t)$ given by
\begin{align*}
\mathbf{Z}^\varepsilon(t) 
= 
\mathbf{Z}^\varepsilon(0) 
\,+ 
&\sum_{r=1}^R 
P_{r}
\left(
\int_0^t \beta_r(\mathbf{Z}^\varepsilon(s))
\,
\widetilde{\lambda}^\varepsilon_r(\llbracket \mathbf{Z}^\varepsilon(s) \rrbracket_{\varepsilon})\,\mbd s\right)
\boldsymbol{\nu}_r
\\
+&
\sum_{r=1}^R 
\boldsymbol{\nu}_r 
\int_0^t (1 - \beta_r(\mathbf{Z}^\varepsilon(s))
\, \widetilde{\lambda}_r^\varepsilon(\mathbf{Z}^\varepsilon(s))\,\mbd s \\ 
+ 
&\sum_{r=1}^R 
\sqrt{\varepsilon} \, \boldsymbol{\nu}_r
\int_0^t 
\sqrt{\displaystyle 
(1 - \beta_r(\mathbf{Z}^\varepsilon(s))
\, \widetilde{\lambda}_r^\varepsilon(\mathbf{Z}^\varepsilon(s))}\, 
\mbd W_r(s),
\end{align*}
where $\lbrace W_r \rbrace_{r=1}^{R}$ and 
$\lbrace P_r\rbrace_{r=1}^R$ 
are standard Wiener and Poisson processes, respectively, all 
mutually independent, and 
$\llbracket \mathbf{x} \rrbracket_{\varepsilon}$ is the closest 
in $\varepsilon\mathbb{L}$ to $\mathbf{x}$, or equivalently 
$\llbracket \mathbf{x} \rrbracket_{\varepsilon} 
= \varepsilon \llbracket \mathbf{x}/\varepsilon \rrbracket$. 
The generator of this process is given by
\begin{align*}
\mathcal{G}^\varepsilon f(\mathbf{x}) 
\,=& 
\frac{1}{\varepsilon}
\sum_{r=1}^R \beta_r(\mathbf{x}) \,
\widetilde{\lambda}_r^\varepsilon
(\llbracket \mathbf{x} \rrbracket_{\varepsilon})
\left[f(\mathbf{x}+\varepsilon\boldsymbol{\nu}_r) - f(\mathbf{x})\right] 
\\
+& \sum_{r=1}^R \left(
1 - \beta_r(\mathbf{x}) \right) \,
\widetilde{\lambda}_r^\varepsilon(\mathbf{x}) 
\, \boldsymbol{\nu}_r\cdot\nabla f(\mathbf{x}) 
\\ +& 
\frac{\varepsilon}{2}
\sum_{r=1}^R \left(
1 - \beta_r(\mathbf{x})
\right)
\, 
\widetilde{\lambda}_r^\varepsilon
(\mathbf{x}) \, (\boldsymbol{\nu}_r\otimes\boldsymbol{\nu}_r)
:\nabla\nabla f(\mathbf{x}).
\end{align*}
We now obtain a weak error estimate between the processes
$\mathbf{X}^\varepsilon(t)$ and $\mathbf{Z}^\varepsilon(t)$ 
in the large volume limit as $\varepsilon \rightarrow 0$. 

\begin{proposition}
Let blending functions satisfy Assumption~$\ref{asblend}$.
Let $g \in C^3(\mathbb{R}^N)$, then there exists a 
constant $C > 0$, independent of $\varepsilon$, such 
that
\begin{equation}
\label{eq:weak_error}
\Big\lvert
\mathbb{E}_{\mathbf{x}}\left[{g(\mathbf{X}^\varepsilon(t))} \right] 
- 
\mathbb{E}_{\mathbf{x}}\left[{g(\mathbf{Z}^\varepsilon(t))} \right]
\Big\rvert 
\leq C \varepsilon^2, \quad \mbox{ for } t \in [0,T],
\end{equation}
where $\mathbf{X}^\varepsilon(0) 
= \mathbf{Z}^\varepsilon(0) 
= \mathbf{x} \in \varepsilon \, \Omega_0$.
\end{proposition}
\begin{proof}
Let 
$u^\varepsilon(t, \mathbf{x}) 
= \mathbb{E}_\mathbf{x}\left[g(\mathbf{X}^\varepsilon(t))\right]$ 
and $v^\varepsilon(t, \mathbf{x}) 
= \mathbb{E}_\mathbf{x}\left[g(\mathbf{Z}^\varepsilon(t))\right]$.  
We wish to obtain a bound on 
$E^\varepsilon(t,\mathbf{x}) 
= u^\varepsilon(t, \mathbf{x}) - v^\varepsilon(t,\mathbf{x})$.  
Then taking the derivative with respect to $t$ and using
the fact that 
$\llbracket \mathbf{x} \rrbracket_{\varepsilon} = \mathbf{x}$
for all $\mathbf{x} \in \varepsilon \Omega_0$, we obtain
\begin{align*}
\partial_t E^\varepsilon(t, \mathbf{x}) 
&- 
\mathcal{G}^\varepsilon E^\varepsilon(t, \mathbf{x}) 
= 
\partial_t u^\varepsilon(t, \mathbf{x}) 
- \mathcal{G}^\varepsilon u^\varepsilon(t, \mathbf{x}) 
= 
\left(
\mathcal{L}^\varepsilon 
- 
\mathcal{G}^\varepsilon
\right) 
u^\varepsilon(t, \mathbf{x})
\\
&= 
\sum_{r=1}^{R} 
(1-\beta_r(\mathbf{x}))\,
\widetilde{\lambda}^\varepsilon_r(\mathbf{x})
\Big[
\frac{u^\varepsilon(t, \mathbf{x}+\varepsilon \boldsymbol{\nu}_r) 
- u^\varepsilon(t, \mathbf{x})}{\varepsilon} 
- \boldsymbol{\nu}_r\cdot\nabla u^\varepsilon(t, \mathbf{x}) 
\\
&\qquad - \frac{\varepsilon}{2}
\, (\boldsymbol{\nu}_r\otimes\boldsymbol{\nu}_r):
\nabla\nabla u^\varepsilon(t,\mathbf{x})\Big].
\end{align*} 
Since $u^\varepsilon(t, \cdot)$ is in $C^3(\mathbb{R}^N)$, 
we can apply Taylor's theorem up to the second order on 
$u^\varepsilon(t, \mathbf{x} + \varepsilon \boldsymbol{\nu}_r)$ 
to obtain
\begin{equation}
\partial_t E^\varepsilon(t, \mathbf{x}) 
- \mathcal{G}^\varepsilon E^\varepsilon(t, \mathbf{x})  
= 
c^\varepsilon(t, \mathbf{x}),
\label{eq:ereq}
\end{equation}
where
$$
c^\varepsilon(t, \mathbf{x})
=
\sum_{r=1}^R
\frac{\varepsilon^2}{6} \,
(1-\beta_r(\mathbf{x}))\,
\widetilde{\lambda}^\varepsilon_r(\mathbf{x})
\,
(\boldsymbol{\nu}_r\otimes\boldsymbol{\nu}_r\otimes\boldsymbol{\nu}_r)
:\nabla\nabla\nabla u^\varepsilon(t, \xi_r),
$$
for some $\xi_r$ lying on the line between $\mathbf{x}$ and 
$\mathbf{x} + \varepsilon \boldsymbol{\nu}_r$. 
From (\ref{eq:ereq}) and the fact 
that $E(0, \mathbf{x}) = 0$, it follows that
$$
E^\varepsilon(t, \mathbf{x}) 
= \int_0^t R_{t-s}^\varepsilon \,
c^\varepsilon(s, \mathbf{x})\,\mbd s,
$$
where $R_t^\varepsilon$ is the semigroup operator corresponding 
to $\mathbf{Z}^\varepsilon(t)$.  Applying the uniform bound
(\ref{eq:uniform_bound}) we thus have that
$$
\left\lvert E^\varepsilon(t, \mathbf{x})\right\rvert 
\leq C\varepsilon^2 
\int_0^t \sum_{r=1}^R
\lVert C_{r,r,r}^\varepsilon(s, \cdot)\rVert_{\infty}\,\mbd s 
\leq   C_1 \varepsilon ^2 T e^{K_1 T},	\quad 
\mbox{ for } t \in [0,T].
$$ 
\end{proof}

\begin{remark}
The remainder term $c^\varepsilon(t,\mathbf{x})$ in equation 
(\ref{eq:ereq}) characterises 
the local error at the point $\mathbf{x}$.  We note that 
it is non-zero only in regions where $\beta_r(\mathbf{x})\neq 1$.  
Intuitively we would expect this to imply that 
$\mathbf{Z}^\varepsilon(t)$ is a superior approximation to 
the standard CLE. However, the global error estimate we derived 
is too coarse to capture the distinction between the two 
diffusion approximations, and thus we have only shown that the 
two approximations are consistent: in that the hybrid scheme 
does no worse than the CLE in the large-volume limit.
\end{remark}

\section{Simulating the Hybrid Model}
\label{sec:algorithm}

\noindent
Equation (\ref{eq:hybrid}) provides a general framework 
for simulating chemical systems which can capture both 
the discrete and continuum nature of a biochemical system. 
Any numerical scheme which can generate realisations 
of a jump-diffusion process with inhomogeneous jump rates 
with deterministic jump sizes can be used to 
simulate (\ref{eq:hybrid}). For illustrative reasons 
we describe three different possible numerical schemes, 
the first based on the Gillespie SSA and the second based 
on the modified Next Reaction Method proposed 
in~\cite{anderson2007modified}, and discussed 
in Section \ref{sec:prelims}. Finally, we describe
an alternative scheme based on adaptive thinning which 
works well for systems with bounded blending regions. 
The main difference between the purely jump case (where 
these algorithms have been applied before) and 
\eqref{eq:hybrid} is the fact that in the blending 
region the propensity functions do not remain constant 
between two consecutive reactions.  For the sake 
of clarity, given the propensities 
$\lambda_1, \lambda_2,\ldots, \lambda_R$ and blending 
functions $\beta_1, \beta_2,\ldots, \beta_R$ define
\begin{equation}
\lambda_j'(\mathbf{x}) 
= \beta_j(\mathbf{x})\lambda_j(\llbracket \mathbf{x} \rrbracket),
\quad \mbox{ and } \quad \lambda_j''(\mathbf{x}) 
= (1 - \beta_j(\mathbf{x}))\lambda_j(\mathbf{x}).
\label{notsim}
\end{equation}
Pseudocodes of each approach are given as Algorithms \ref{alg:hybrid},
\ref{alg:hybrid1} and \ref{alg:hybrid3}. They all
have the same input, namely propensities 
$\lambda_1,\lambda_{2}, \ldots, \lambda_R$, 
blending functions $\beta_1,\beta_{2}, \ldots, \beta_R$, 
the stoichiometric matrix 
$\underline{\boldsymbol{\nu}} 
= (\boldsymbol{\nu_1}, \ldots, \boldsymbol{\nu}_R)$, 
the final time of simulation $T$, time steps for the 
CLE $\Delta t$ and $\delta t$ (here, $\delta t \ge \Delta t$)
and initial state $\mathbf{X}(0) \in \mathbb{N}^N$.

\renewcommand{\floatpagefraction}{.7}%

\begin{algorithm}
\SetAlgoLined
Set $t = 0$.
\\
\While{$t<T$}{
\uIf{$\max_{j} \beta_j(\mathbf{X}(t)) = 0$}
{
Simulate CLE \eqref{eq:cle} up to time $t+\delta t$.\\ 
Set $t = t + \delta t$.
}
\uElseIf{$\min_{j} \beta_j(\mathbf{X}(t)) = 1$}{
Compute $\lambda_0 = \sum_{j=1}^R \lambda_j'(\mathbf{X}(t))$.\\
Sample $\tau \sim -\log(u)/\lambda_0$ where $u \sim U[0,1]$.\\
Choose the next reaction $r$ with probability
$\lambda_r'(\mathbf{X}(t))/\lambda_0$.\\
Set $\mathbf{X}(t+\tau) = \mathbf{X}(t) + \boldsymbol{\nu}_r$.\\
Set $t = t + \tau$.
}
\uElse{Compute $\lambda'_{0}=\sum_{j=1}^{R}\lambda'_j(\mathbf{X}(t))$.\\
Sample $\tau \sim -\log{u}/\lambda'_{0}$, where $u\sim U[0,1]$. \\
Choose the next reaction $r$ with probability
$\lambda'_r( \mathbf{X}(t))/\lambda'_{0}$. \\
\eIf{$\tau < \Delta t$}{Simulate  CLE \eqref{eq:cle1} 
up to time $t+\tau$ and set 
$\mathbf{X}(t + \tau) = \mathbf{X}(t+\tau) + \boldsymbol{\nu}_r$. \\
Set $t = t + \tau$.}
{Simulate CLE \eqref{eq:cle1} up to time $t+\Delta t$. \\
Set $t = t + \Delta t$.
}}}
\caption{{\it generating approximate realisations of  
hybrid model $(\ref{eq:hybrid})$.
\label{alg:hybrid}}} 
\end{algorithm}

\subsection{Hybrid simulations based on the Gillespie SSA}

\noindent
The steps to simulate the jump-diffusion process 
\eqref{eq:hybrid} based on an extension of the Gillespie 
SSA are described in Algorithm \ref{alg:hybrid}. As we 
can see in regions where $\beta_{r}(x)$ are $1$, the 
scheme reduces to the standard Gillespie SSA, and thus 
simulates the discrete dynamics exactly. Analogously, 
in regions where $\beta_{r}$ are all zero one can use 
a larger time-step $\delta t$ to evolve CLE \eqref{eq:cle} 
since it is not necessary to approximate the solutions 
of \eqref{eq:hybrid} in such regions, which can only be 
done with $\mathcal{O}(\Delta t)$ accuracy. 
In the intermediate regime for the continuous part 
of the dynamics the following CLE is used
\begin{equation}
\label{eq:cle1}
\mbd \mathbf{X}(t) 
= 
\sum_{j=1}^{R} \lambda_j''(\mathbf{X}(s)) \, \boldsymbol{\nu}_j\,\mbd t 
+ 
\boldsymbol{\nu}_j\sqrt{\lambda_j''(\mathbf{X}(s))}\,\mbd W_j(s).
\end{equation}
If during the CLE time step $[t,t+\Delta t)$ a discrete event 
is occurring at time $t+\tau$ we simulate the CLE up to that 
time and then add the discrete event.

\noindent
To simulate the diffusion part of the hybrid scheme we 
make use of the weak trapezoidal method described 
in~\cite{anderson2011weak}. Given the current 
state $\mathbf{X}^n$ we perform the following 
two steps to obtain $\mathbf{X}^{n+1}$: 
\begin{eqnarray*}
\mathbf{X}^* &=& 
\mathbf{X}^{n}  
+ \frac{\Delta t}{2}\sum_{j=1}^R\boldsymbol{\nu}_j
\lambda''_j\left(\mathbf{X}^{n}\right) 
+ \sqrt{\frac{\Delta t}{2}}\sum_{j=1}^R \boldsymbol{\nu}_j
\lambda''_j(\mathbf{X}^{n})\, \xi_j, \\
\mathbf{X}^{n+1} &=& 
\mathbf{X}^* 
+ 
\frac{\Delta t}{2}\sum_{j=1}^R h_j(\mathbf{X}^*,\mathbf{X}^{n})
+
\sqrt{
\displaystyle
\frac{\Delta t}{2}}\sum_{j=1}^R \sqrt{
\left[
\displaystyle
h_j(\mathbf{X}^*,\mathbf{X}^{n})
\right]^{+}} \boldsymbol{\nu}_j \,\xi'_j,
\end{eqnarray*}
where $h_j(\mathbf{X}^*,\mathbf{X}^{n})
=
2 \lambda''_j(\mathbf{X}^*) -  \lambda''_j (\mathbf{X}^{n}),$
$[a]^{+} = \max\left(0, a\right)$ and 
$\xi_1, \xi_2, \ldots, \xi_R$, 
$\xi'_1, \xi'_2, \ldots \xi'_R$ are mutually 
independent standard Gaussian random variables.

\subsection{Hybrid simulations based on the Next Reaction Method}

\noindent
A second algorithm, based on \cite{anderson2007modified}, for 
simulating \eqref{eq:hybrid} is described in Algorithm~\ref{alg:hybrid1}.
While it is entirely equivalent to the standard Next Reaction Method, 
the modified scheme keeps explicit track of the internal times $T_k$ 
and the next firing time $F_k$ of each Poisson process $P_k$ which 
simplifies integrating diffusion steps into the scheme. We note 
that the hybrid scheme described in \cite{ganguly2014jump} also 
employs a similar discretisation. Again in the presence of diffusion, 
it is no longer true that the propensity $\lambda'_j(\mathbf{X}(s))$ 
is constant from $t$ until the next reaction. Computing the next 
reaction time is equivalent to solving the following first passage 
time problem:
\begin{align*}
\mbox{Compute }& 
\inf\lbrace s \geq t \, : \, 
T_r(s) = F_r(t) - \log u  \mbox{ for some } 
r =1,2, \ldots, R\rbrace \mbox{ where:}\\
\mbd \mathbf{X}(s) 
&= \sum_{j=1}^{R}\lambda_j''(\mathbf{X}(s))\boldsymbol{\nu}_j\,\mbd t 
+ \boldsymbol{\nu}_j\sqrt{\lambda_j''(\mathbf{X}(s))}\,\mbd W_j(s)\\
\mbd T_r(s) &= \lambda'_r(\mathbf{X}(s))\,\mbd s, \quad r=1,2,\ldots, R.,
\end{align*}
In Algorithm~\ref{alg:hybrid1}, we use an Euler discretization 
of $T_r(s)$.

\begin{algorithm}
\SetAlgoLined
\BlankLine
Generate $R$ independent, $U[0,1]$ random numbers 
$u_j$, $j=1,2,\dots,R$.\\
Set $t = 0$. Set $F_j = -\log(u_j)$ and $T_j = 0$, for each 
$j=1,2,\dots,R$.\\
Compute the weighted propensities 
$\lambda'_j = \beta(\mathbf{X}(0))\lambda_j(\mathbf{X}(0))$.\\
\While{$t<T$}{
\eIf{$\max_{j}\beta_j(\mathbf{X}(t)) = 0$}
{Simulate CLE (\ref{eq:cle}) up to time $t + \delta t$.
Set $\tau = \delta t$.}
{
Set $\tau_j = (F_j - T_j)/\lambda'_j$ for $j=1,2,\ldots,R.$\\
Let $r = \mbox{argmin}_{j}\lbrace \tau_j \rbrace$ 
and set $\tau = \tau_{r}$. \\
\eIf{$\min_{j}\beta_j(\mathbf{X}(t)) = 1$}
{
Set $\mathbf{X}(t + \tau) = \mathbf{X}(t+\tau) 
+ \boldsymbol{\nu}_r$.\\
Sample $u \sim U[0,1]$ and set 
$F_r = F_r - \log(u)$.}
{
\eIf{$\tau < \Delta t$}{
Simulate CLE \eqref{eq:cle1} up to time $t+\tau$.\\  
Set $\mathbf{X}(t + \tau) 
= \mathbf{X}(t+\tau) + \boldsymbol{\nu}_r$.\\
Sample $u \sim U[0,1]$
and set $F_r = F_r - \log(u)$.
}{
Simulate CLE \eqref{eq:cle1} up to time $t+\Delta t$.
Set $\tau = \Delta t$.
}
}
Set $t = t + \tau$. Set $T_j = T_j + \lambda'_j \tau$
for $j=1,2,\dots,R$.\\
Update the  propensities 
$\lambda_j' 
= 
\beta_j(\mathbf{X}(t))\lambda_j(\llbracket \mathbf{X}(t)\rrbracket).$
}}
\caption{{\it generating approximate realisations of  
hybrid model $(\ref{eq:hybrid})$.
\label{alg:hybrid1}}}
\end{algorithm}

\subsection{Thinning Method}

\noindent
In the special case where one can bound the value of the weighted 
propensities $\lbrace \lambda_r'\rbrace_{r=1}^R$, it is possible 
to use a third method, based on standard thinning methods 
for sampling inhomogeneous Poisson processes, see for 
example~\cite{luc1986non,lewis1978simulation} and more 
recently~\cite{thanh2015simulation}. The main advantage 
of this scheme would be that it avoids the necessity to 
approximate the solution of the first passage time problem 
associated with the modified Next Reaction Method, 
thus being potentially more efficient. While one can find such 
bounds for many chemical systems, the added caveat is that 
the bounds must be known a priori, and choosing them too 
loosely will severely degrade the performance of the scheme.

To this end, we shall assume that there exist constants 
$\Lambda_1,\Lambda_2,\ldots, \Lambda_R$, where
\begin{equation}
\label{eq:thinning_boundedness_condition}
\lambda'_r(\mathbf{x}) \leq \Lambda_r,	\quad \mbox{ for all } \mathbf{x}.
\end{equation}
These constants will form the additional input of 
Algorithm~\ref{alg:hybrid3}. Suppose that 
$\beta_r(\mathbf{X}(t)) > 0$ for some $r$. 
To compute the next jump time of the $r$-th reaction, 
we sample from a dominating homogeneous process with 
rate $\Lambda_r$, so that the next jump time for 
the $r$-th reaction occurs at time $t + \tau_r$ where
$$
\tau_r = -\frac{\log(u)}{\Lambda_r},\quad u \sim U[0,1].
$$
Suppose that $t + \tau_{r}$ is the first jump occurring 
after the current time $t$. To determine whether the 
reaction $r$ will occur at time $t + \tau_r$, we 
sample $u' \sim U[0,1]$ and perform the reaction only if
\begin{equation}
\label{eq:thinning_condition}
\Lambda_r u' \leq \lambda'_r(\mathbf{X}(t + \tau_r)),
\end{equation}
where $\mathbf{X}(t + \tau_r)$ is the state of the 
process after simulating the Langevin dynamics 
from time $t$ to $t + \tau_r$. This thinning approach
can be integrated into the Gillespie SSA, demonstrated in  
Algorithm \ref{alg:hybrid3}.  At each timestep, three cases 
can occur. If $\min \beta_r(\mathbf{X}(t)) = 1$ (i.e. 
the process $\mathbf{X}(t)$ is a pure jump process) 
we use the standard Gillespie SSA.  
If $\max \beta_r(\mathbf{X}(t)) = 0$ (i.e. the process 
is purely diffusive), then we perform a ``macro-step'' 
of the CLE dynamics of size $\delta t$.  The final case 
is where there is both diffusion and jumps, we simulate 
the homogeneous dominating process with rate 
$\Lambda_0 = \sum_{r=1}^R\Lambda_r$, and accept/reject 
according to condition~(\ref{eq:thinning_condition}).  

The main advantage is that we avoid the error arising from 
the piecewise constant approximation of integral $T_r(t)$.  
In particular, one can use higher order methods for simulating 
the Langevin dynamics within the blending region to obtain 
a better weak order of convergence in $\Delta t$.  
The drawback of this approach is the necessity to know 
\emph{a priori} the upper bounds $\Lambda_r$, assuming such 
bounds exist. Care must be taken so that the bounds are 
not too pessimistic, otherwise the dominating homogeneous 
Poisson process will fire very rapidly when the system lies 
within a blending region. In such cases the bounds can be 
tuned by running exploratory simulations and keeping track 
of the acceptance rate for each reaction.

\begin{algorithm}
\SetAlgoLined
\BlankLine
Set $t = 0$.\\
\While{$t<T$}{
\uIf{$\max_{r=1,2,\ldots,R}\beta_r(\mathbf{X}(t)) = 0$}
{Simulate CLE \eqref{eq:cle} up to time $t+\delta t$.\\
Set $t = t + \delta t$.
}\uElseIf{$\min_{r=1,2,\ldots,R} \beta_r(\mathbf{X}(t)) = 1$}
{Let $\lambda_0 = \sum_{r=1}^R\lambda_r'(\mathbf{X}(t))$.\\
Let $u \sim U[0,1]$ and $\tau = -\log(u)/\lambda_0$.\\
Choose index $r$ with probability 
$\lambda'_r(\mathbf{X}(t))/\lambda_0$.\\
Set $\mathbf{X}(t+\tau) = \mathbf{X}(t) + \boldsymbol{\nu}_r$.\\
Set $t = t + \tau$.
}\uElse{
Let $\Lambda_0 = \sum_{r=1}^R \Lambda_r$.\\
Let $u \sim U[0,1]$ and set $\tau = -\log(u)/\Lambda_0$.\\
Simulate CLE (\ref{eq:cle1}) up to time $t + \tau$ using stepsize $\Delta t$.\\
Let $u' \sim U[0,1]$.\\
\If {$\Lambda_0 u' \leq \sum_{r=1}^R \lambda'_r(\mathbf{X}(t+\tau))$}{
Let $r$ be smallest index such that 
$\Lambda_0 u' \leq \sum_{j=1}^{r}\lambda_j'(\mathbf{X}(t+\tau)).$\\
Set $\mathbf{X}(t+\tau) = \mathbf{X}(t) + \boldsymbol{\nu}_r$.\\
Set $t = t + \tau$. }
}}
\caption{{\it generating approximate realisations of  
hybrid model $(\ref{eq:hybrid})$.
\label{alg:hybrid3}}}
\end{algorithm} 

\section{Numerical Investigations}

\noindent
We illustrate the main features of the hybrid framework 
described in Section \ref{subsec:stoc_fram} and demonstrate 
the use of Algorithms \ref{alg:hybrid}, \ref{alg:hybrid1} 
and \ref{alg:hybrid3} by considering three numerical examples.  
Each of these examples were implemented in the programming 
language {\sc Julia}~\cite{bezanson2012julia}. 

\subsection{Lotka-Volterra Model}
\label{sec:lotka_volterra}

\noindent
As a first example, we consider a chemical system consisting 
of two reacting chemical species $A$ and $B$ undergoing 
the following reactions
\begin{equation}
\begin{aligned} \label{eq:lotka}
&A  \xrightarrow{k_1} 2A, \qquad\qquad 
A + B \xrightarrow{k_2} 2B, \qquad\qquad
&B \xrightarrow{k_3} \emptyset.
\end{aligned}
\end{equation}
The chemicals $A$ and $B$ can be considered to be in a 
``predator-prey" relationship with $A$ and $B$ as 
prey and predator, respectively.  The reaction-rate 
equations corresponding to reactions (\ref{eq:lotka}) 
would then be the standard Lotka-Volterra model.  
We choose the dimensionless parameters $k_1 = 2.0,$  $k_2 = 0.002$ 
and $k_3 = 2.0.$ The initial condition is chosen to be 
$A(0)=50$ and $B(0)=60$.  
A histogram generated from $10^3$ independent SSA simulations 
of this system up to time $T = 5$ is shown in 
Figure~\ref{fig:lotka1}. The dashed line depicts the evolution 
of the deterministic reaction rate equation starting from the 
same initial point. One sees that the nonequilibrium dynamics 
force the system to spend time in both low  and high concentration 
regimes. Due to the time spent in states with high propensity, 
the SSA is computationally expensive to simulate. It is clear 
that away from the boundary, using an approximation such as the 
CLE would be computationally beneficial. The CLE corresponding 
to (\ref{eq:lotka}), choosing the multiplicative noise as described 
in \cite{gillespie2000chemical}, is given by
\begin{eqnarray*}
\mbd A(t)
 \!\!\!\! &=& \!\!\!\! 
\Big( k_1A(t) - k_2 A(t) B(t) \Big) \,\mbd t  
\! +\!  \sqrt{k_1 A(t)} \, \mbd W_1(t)
 \! - \! \sqrt{k_2 A(t)B(t)} \, \mbd W_2(t), \\
\mbd B(t)
 \!\!\!\! &=& \!\!\!\!  
\Big( k_2 A(t) B(t) - k_3 B(t) \Big)\,\mbd t  
\! - \! \sqrt{k_2 A(t)B(t)} \, \mbd W_2(t)
\! - \! \sqrt{k_3 B(t)} \, \mbd W_3(t),
\end{eqnarray*}
where $W_1(t)$, $W_2(t)$ and $W_3(t)$, 
are three standard independent Brownian motions. 
For a non-negative initial condition, the process 
$(A(t), B(t))$ will remain nonnegative, however, 
this will not be the case for fixed-timestep discretisation.  
In particular, an Euler discretisation ($A_n, B_n)$ will 
contain a term of the form $-\sqrt{k_3 B_n \Delta t} \, \xi$, 
where $\xi$ is a standard Gaussian random variable, 
which can cause the discretised process to cross the 
$B = 0$ axis if the process sufficiently close to this line. 
Thus, it is essential that reflective boundary conditions 
are imposed to ensure positivity. 
However, even if positivity is guaranteed, there is no 
reason to believe that the CLE will correctly approximate 
the dynamics near the axes. This motivates the use of the 
hybrid model to efficiently simulate this chemical system.

To simulate the hybrid model we use Algorithm \ref{alg:hybrid}, 
choosing blending functions $\beta_1$, $\beta_2$ and $\beta_3$ 
as described in (\ref{eq:multi_species_blending}). 
We simulate the Langevin dynamics in the blending region 
using a timestep of size $\Delta t = 10^{-3}$, and 
a timestep $\delta t = 10^{-2}$ outside the blending regions. 
In Figure \ref{fig:lotka_timedep}(a), we plot  $\mathbb{E}(A(t))$ 
at times $T = 1, 2, \ldots, 6$,  generated from $10^4$ independent 
realisations of each model. Error bars denote $95\%$ confidence 
intervals. The hybrid models were simulated for different values 
of $I_1^i$ and $I_2^i$, $i=1,2,3$, however, the results were 
not plotted as the Monte Carlo error was too large to distinguish 
between the schemes.  The hybrid scheme displayed in 
Figure~\ref{fig:lotka_timedep}(a) has blending regions with 
parameters $I_1^i = 25.0$ and $I_2^i = 35.0$.  While all models 
agree approximately at small times, 
for larger $T$ the averages generated from SSA and CLE differ 
significantly. Indeed, for $T=5$ (corresponding to a single 
period of the deterministic system) the means from the SSA 
and CLE differ by three orders of magnitude. On the other hand, 
the hybrid scheme remains in good agreement with the SSA. 
In Figure~\ref{fig:lotka_timedep}(b) we compare the average 
computational (CPU) time to simulate each model up to time $T$, 
averaged over $10^4$ realisations. This was measured in seconds, 
using the standard {\sc Julia} functions {\sc tic} and {\sc toc}. 
The computational cost of the hybrid scheme was plotted for 
three different choices of blending regions namely 
$(I_1^i, I_2^i) = (5, 15)$, $(10, 25)$ and $(25, 35)$, respectively. 
For small $T$ the SSA and hybrid schemes require a comparable amount 
of computational effort.  However, as $T$ increases, the computational 
cost of the SSA scheme dramatically increases, while the cost of 
the hybrid scheme remains approximately constant.  To compute the 
average value at time $T = 6$, the hybrid scheme was on average $2$ 
orders of magnitude cheaper to run.  As expected, the computational 
effort is smaller when the blending region is closer to the boundary.  
However, relative to the computational cost of the SSA and CLE, 
varying the blending region does not significantly alter performance. 
For these simulations, the value of $\Delta t$ and $\delta t$ where 
chosen manually by computing the error for a number of short 
exploratory runs. A more sophisticated implementation 
of the hybrid model would require an adaptive scheme for the 
Langevin part of the process.
 
\begin{figure}[t]
\centering
\includegraphics[width=0.75\textwidth]{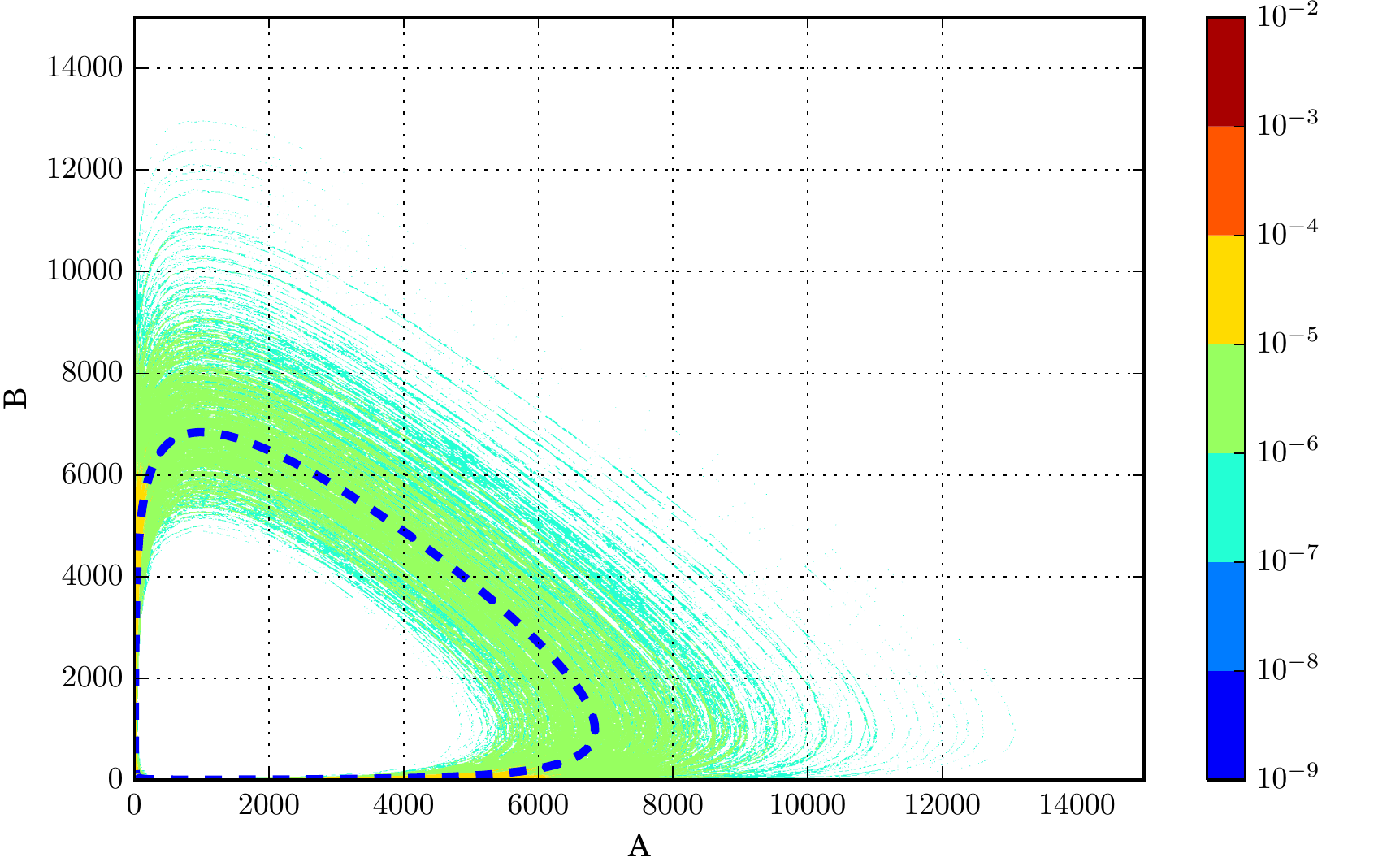}
\caption{{\it Histogram of $10^3$ SSA simulations of 
$\ref{eq:lotka}$ up to time $T = 5.0$ starting from 
$A(0)=50$ and $B(0)=60$.  The dashed line is the solution 
of the corresponding deterministic reaction rate equations.}}
\label{fig:lotka1}
\end{figure}

\begin{figure}[t]
\centerline{
\includegraphics[width=0.46\textwidth]{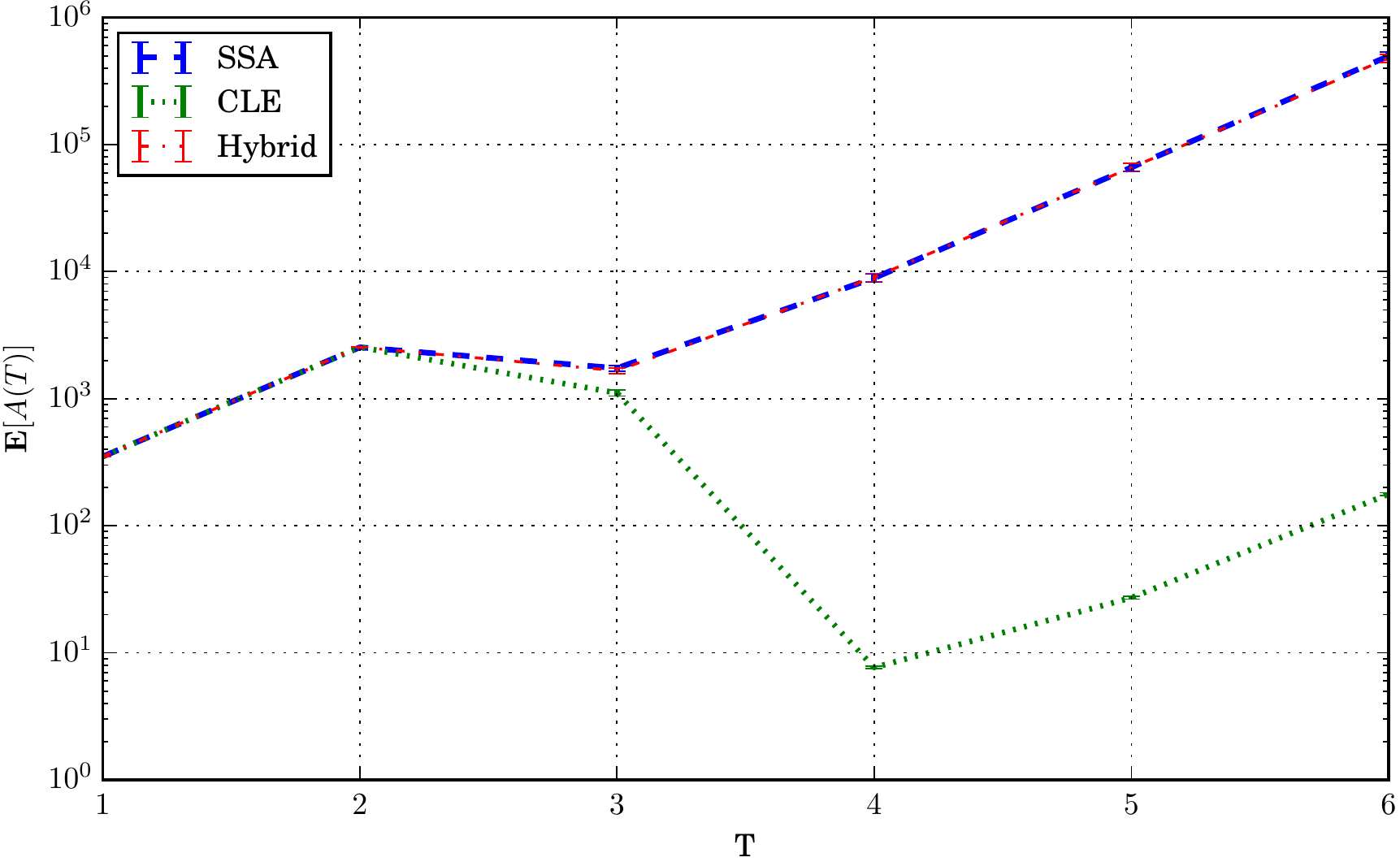}
\hskip 6mm
\includegraphics[width=0.46\textwidth]{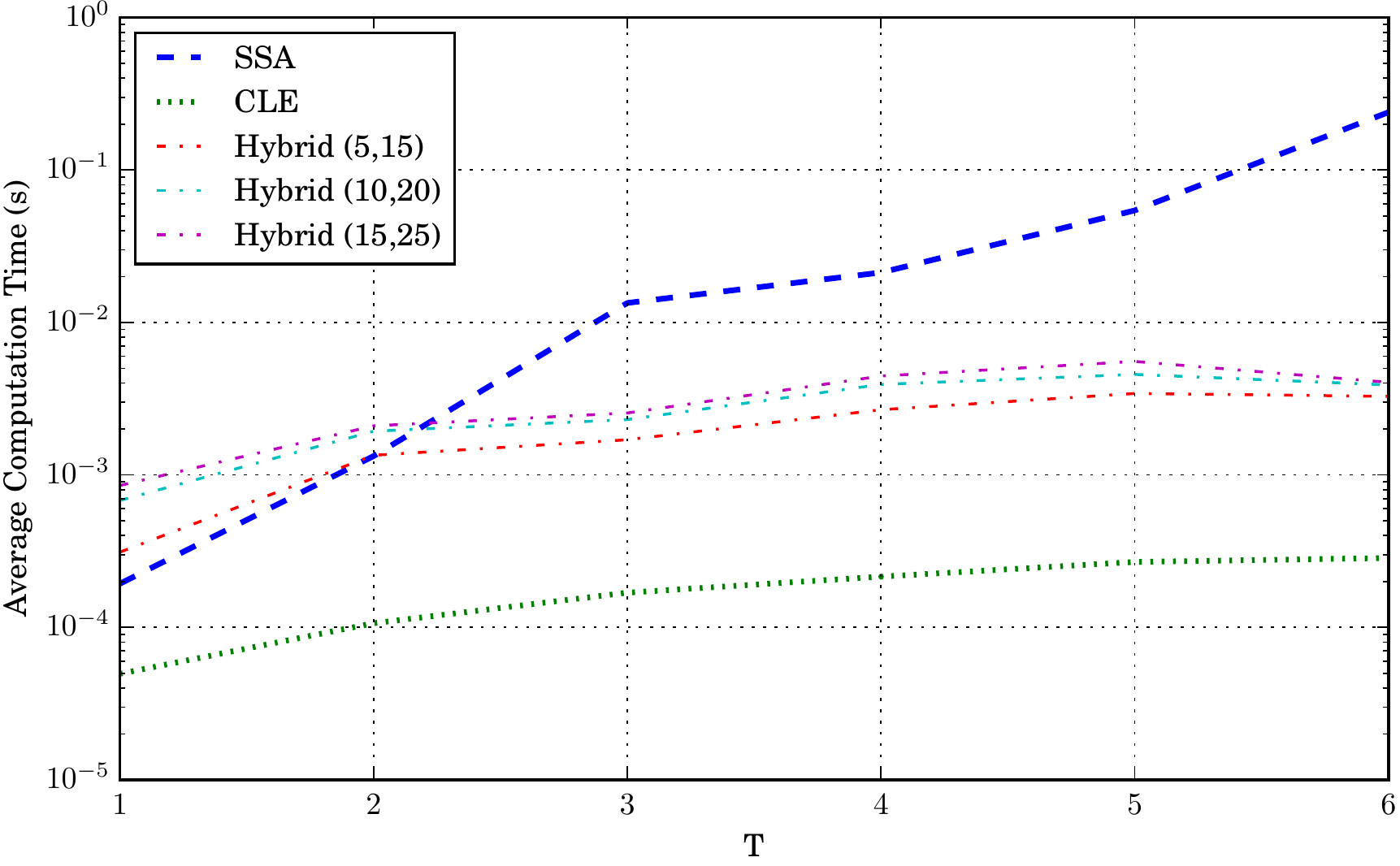}
}
\vskip -4.6cm
\leftline{\hskip 2mm (a) \hskip 6.2cm (b)}
\vskip 4cm
\caption{(a) 
{\it Plots of $\mathbb{E}(A(T))$ as a function of $T$ for the SSA, 
CLE and hybrid schemes;}
(b) {\it the corresponding average computational (CPU) time in seconds
as a function of $T$.}}
\label{fig:lotka_timedep}
\end{figure}

\subsection{Steady State Simulations}
\label{sec:dimer_steady_state}

\noindent
As a second example, we consider a chemical system consisting of 
two species $A$ and $B$ in a reactor of volume $V$. 
The species are subject to the following system of four chemical 
reactions~\cite{Erban:2007:PGS}:
\begin{equation}
\label{eq:dimer}
A + A \xrightarrow{k_1} \emptyset, 
\qquad 
A + B \xrightarrow{k_2} \emptyset,
\qquad
\emptyset \xrightarrow{k_3} A, 
\qquad 
\emptyset \xrightarrow{k_4} B.
\end{equation}
This corresponds to a jump process $\mathbf{X}(t)$ having 
stoichiometric vectors 
$$
\boldsymbol{\nu}_1 = (-2,0)^\top,
\quad \boldsymbol{\nu}_2  = (-1, -1)^\top,
\quad \boldsymbol{\nu}_3 = (1, 0)^\top,	
\quad \boldsymbol{\nu}_4 = (0,1)^\top,
$$
with corresponding propensities (depending on the volume $V$):
$$	
\lambda_1(a, b) = \frac{k_1a(a-1)}{V},	
\quad 
\lambda_2(a, b) = \frac{k_2 a b}{V},
\quad
\lambda_3(a, b) = k_3 V,
\quad 
\lambda_4(a,b) = k_4 V.
$$
The dimensionless reaction rates are given by $k_1 = 10^{-3}$, 
$k_2 = 10^{-2}$, $k_3 = 1.2$ and $k_4 = 1$.
As a first numerical experiment, we compute the evolution 
of the distribution of $(A,B)$ over time. We assume 
that $V = 0.25$, and that the initial distribution is 
a  ``discrete'' Gaussian mixture, namely the Gaussian 
mixture
\begin{equation}
\label{eq:dimer_initial}
\rho_0 = \mathcal{N}((30, 10), 1) + \mathcal{N}((20, 30), 1),
\end{equation}
restricted to the lattice $\mathbb{N}^2$. For each scheme 
(SSA, CLE and hybrid), the distribution is approximated by 
a histogram generated from $10^7$ independent realisations 
of the process. The CLE was simulated using the weak second 
order trapezoidal scheme described in \cite{anderson2011weak}. 
To ensure positivity of the CLE, reflective boundary conditions 
were imposed at the boundary of the  orthant. The hybrid scheme 
was simulated the hybrid next reaction scheme detailed in 
Algorithm \ref{alg:hybrid1}. The timestep was chosen to 
be $\Delta t = \delta t = 0.1$. The blending functions for 
the hybrid scheme were chosen according 
to (\ref{eq:multi_species_blending}), with $I_1^1 = I_1^2 = 5$ 
and $I_2^1 = I_2^1 = 10$.  In Figure~\ref{fig:dimer_timedep} 
we plot the distribution approximated using each scheme at 
times $t = 0, 1, 10$ and $100$. As expected, when the 
concentrations of $A$ and $B$ remain abundant, all three models 
agree. As the distribution approaches the low concentration 
regions, the discrete nature of the chemical system becomes 
important, and the CLE is no longer able to correctly capture 
the dynamics. Indeed, at time $100$ one observes a significant 
difference between the SSA and CLE distributions. On the other 
hand, the hybrid scheme provides a good approximation to the 
SSA at all times, but benefiting from a computational advantage 
in the large concentration regimes.

\begin{figure} [htp]
\centering
 \includegraphics[width=\textwidth]{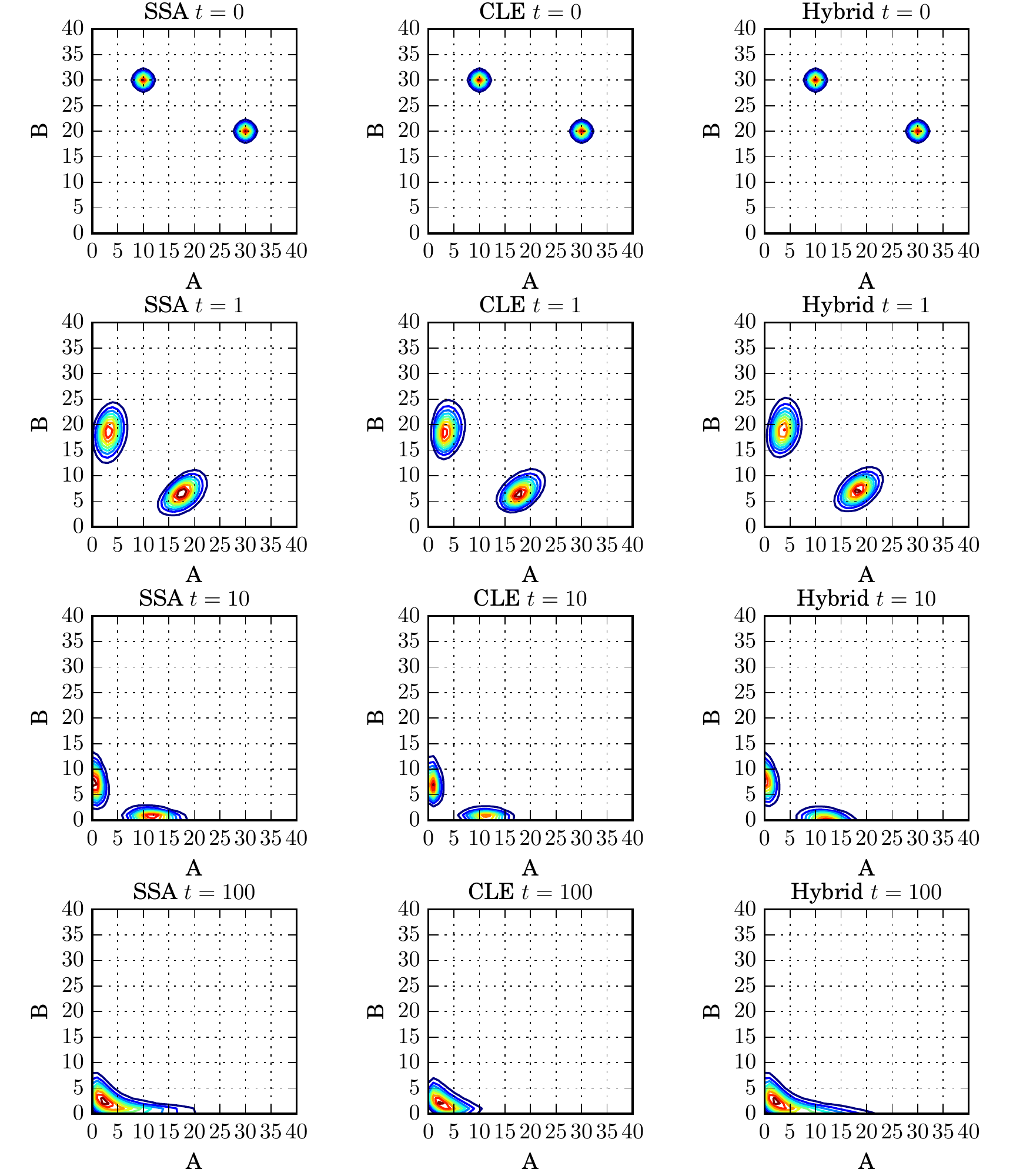}
\caption{{\it Time evolution of the biochemical system 
$(\ref{eq:dimer})$ using SSA, CLE and the hybrid scheme starting 
from initial distribution $\rho_0$ given by 
$(\ref{eq:dimer_initial})$.}}
\label{fig:dimer_timedep}
\end{figure}

The corresponding Markov jump process $\mathbf{X}(t)$ can be 
shown to possess a unique stationary steady 
state~\cite{Erban:2007:PGS}. We use all three models to compute 
the first two moments $M_1$ and $M_2$ of the stationary 
distribution, for decreasing values of $V$. The moments 
were approximated using ergodic average of the discretised 
schemes, i.e.
$$
M_1 \approx \frac{1}{T}\sum_{t_i \leq T}(t_{i+1} - t_{i}) A_{i} 
\mbox{ and } 
M_2 \approx \frac{1}{T}
\sum_{t_i \leq T}(t_{i+1} - t_{i})  \left(A_{i}\right)^2,
$$
where $(A_i,B_i)$ is the value of the discretised process at 
time $t_i$ and $0 < t_1 < t_2 < \ldots < t_N = T$ are the jump 
times of the process. Each process was simulated up to 
time $T = 10^7$. For the hybrid and CLE schemes a timestep 
of $\Delta t = 0.1$ was used throughout. The blending region 
was chosen as in the previous example. The first and second 
moment are plotted in Figure \ref{fig:dimer_stat} for 
$V = 2^{-i}$ where $i=0,1,2\ldots, 8$.  While there is 
good agreement between all three schemes for $V$ large, 
the CLE consistently overestimates the moments when $V$ 
is small. On the other hand, the hybrid scheme remains 
robust to this rescaling.

\begin{figure}[t]
\centerline{
\includegraphics[width=0.46\textwidth]{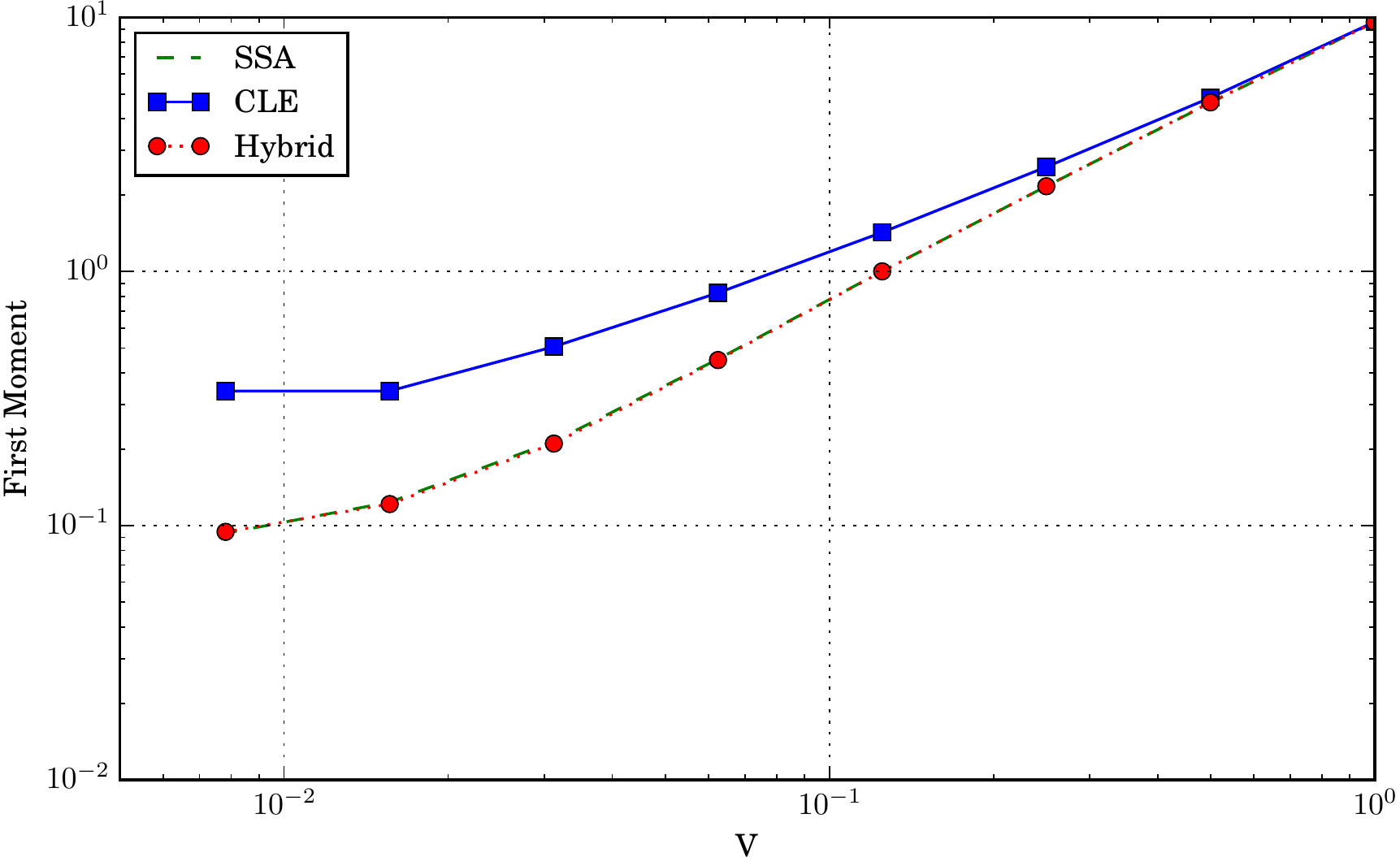}
\hskip 6mm
\includegraphics[width=0.46\textwidth]{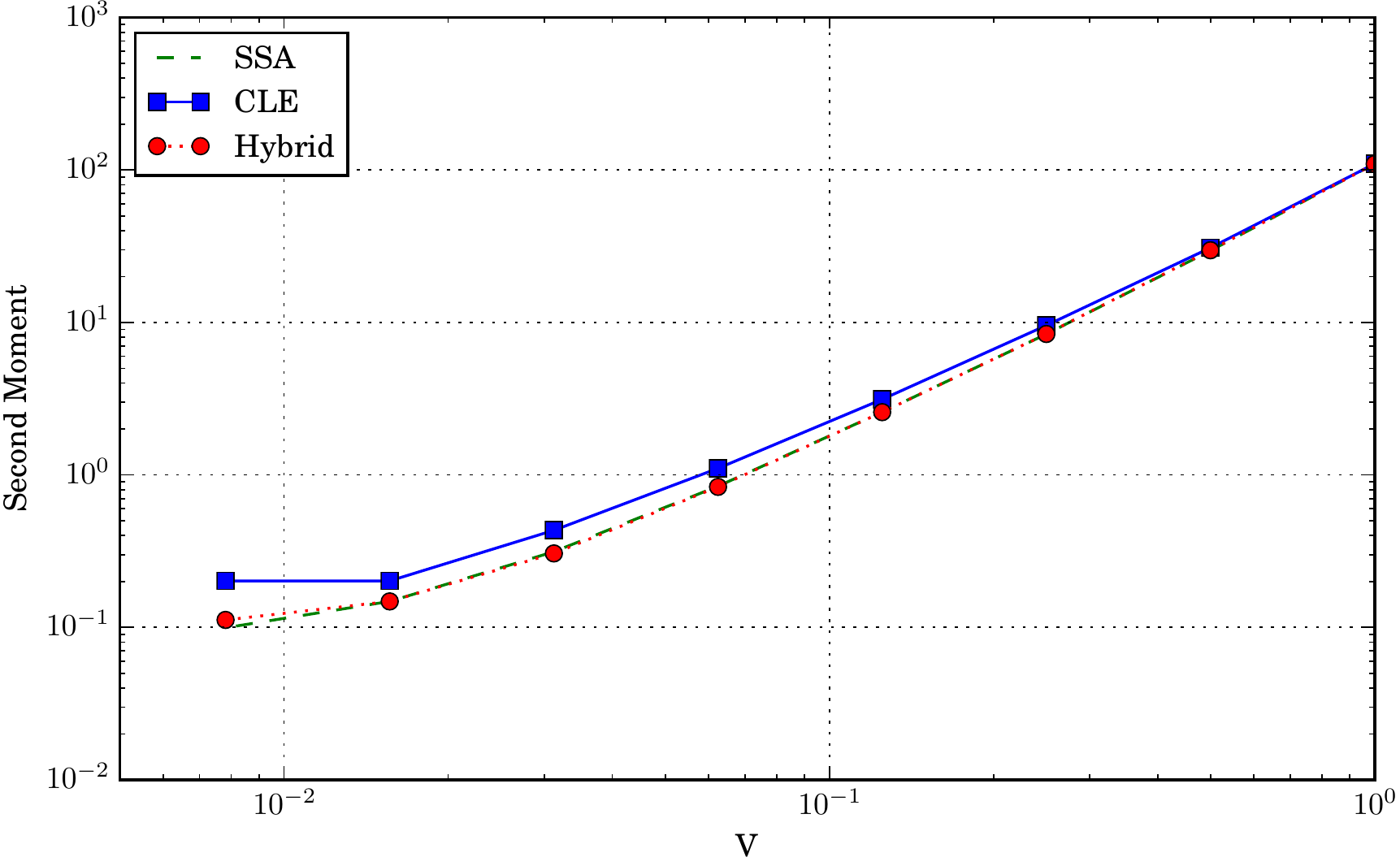}
}
\vskip -4.6cm
\leftline{\hskip 2mm (a) \hskip 6.2cm (b)}
\vskip 4cm
\caption{(a) 
{\it Comparison of the first moments for the three different 
schemes (SSA, CLE and hybrid) for simulation of the biochemical 
system $(\ref{eq:dimer})$ with parameters as in 
Figure~$\ref{fig:dimer_timedep}$ for $V = 2^{-i}$ for $i=0,1,2\ldots, 8$.
}
(b) {\it Comparison of the second moments as a function of $V$.}}    
\label{fig:dimer_stat}
\end{figure}

\subsection{Exit time calculation for the birth-death problem}
\label{sec:exit_time}

\noindent
As a final example, we consider the problem of computing 
the mean extinction time (MET) for a one-dimensional 
birth-death process, namely a system two reactions for one chemical
species $A$:
\begin{equation}
\label{eq:bd}
\begin{aligned}
A \rightarrow \emptyset &\mbox{ with propensity } \lambda_1(n) = k_1 n,\\
\emptyset \rightarrow A &\mbox{ with propensity } \lambda_2(n) = k_2  
\mbox{ if } 0 < n < K, \mbox{ otherwise } 0.
\end{aligned}
\end{equation}
By assuming that $\lambda_2(n) = 0$ for all $n \geq K$, 
the state of the system lies within the finite domain 
$\lbrace 0, 1, 2,\ldots, K \rbrace$. For $k_1 > 0$, the 
birth-death process will hit the extinction state 
$x = 0$ with probability $1$.  We denote by $\mbox{MET}(n)$ 
the MET of the process starting from $A(0)=n$.  
Following directly the approach 
of~\cite[Section 2.1]{doering2005extinction} 
(also see~\cite{Gardiner:1985:HSM,vanKampen:2007:SPP}), 
we obtain
\begin{equation}
\label{eq:bd_mfpt}
\mbox{MET}(n) 
= 
\frac{1}{k_1}
\sum_{m=1}^{n}
\sum_{j=0}^{K-m+1}
\left( \frac{k_1}{k_2} \right)^j
\frac{(m-1)!}{(j+m-1)!},
\qquad
\mbox{ for } n = 1,2, \dots, K.
\end{equation}
The corresponding CLE is given by
\begin{equation}
\label{eq:cle_bd}
\mbd
Y(t) 
= (k_1 - k_2 \, Y(t))\,\mbd t 
+ \sqrt{k_1}\, \mbd W_1(t) - \sqrt{k_2 \, Y(t)}\,\mbd W_2(t),
\end{equation}
for standard independent Brownian motions $W_1(t)$ and 
$W_2(t)$.  The mean first time of $Y(t)$ reaching 
$0$ starting from $Y(0) = x$ can be calculated explicitly 
as
\begin{equation}
\label{eq:bd_cle_mfpt}
2\int_0^x e^{-\Phi(y)}\int_{y}^{K}\frac{e^{\Phi(z)}}{k_2 + k_1 z}\,dz\,\mbd y,
\end{equation}
where $\Phi$ is the potential
$$
\Phi(x) 
= 
\frac{4 k_2}{k_1} \log \left(1 + \frac{k_1 x}{k_2} \right) - 2  x.
$$
Since $n_e = k_2/k_1$ is a unique solution of $\lambda_1(n) = \lambda_2(n)$,
the stochastic birth-death process will fluctuate around 
$n_e$ for a long time before eventually going extinct.  
In general, we expect that any approximation which 
correctly describes the extinction time behaviour of 
the birth-death process must accurately capture the 
behaviour of the process particularly near $n = 0$ 
and $n = n_e$ (and possibly all points in between).  
If $n_e$ is large, then a Gaussian approximation 
(e.g. CLE or the system size expansion) would accurately 
capture the fluctuations around the quasi-equilibrium. 
However, as observed in \cite{doering2005extinction} 
such approximations would suffer close to $n = 0$. 
This suggests that the hybrid scheme with a blending 
region supported between $n = 0$ and $n = n_e$ would be 
a good candidate for an approximation to the process. 
Similar observations have been in more general chemical 
systems~\cite{Hinch:2005:EST,hanggi1984bistable}, where 
it is observed that diffusion approximations of jump 
processes are not able to correctly capture rare events, 
even when the system is in a regime where the CLE correctly 
captures both the transient and stationary dynamics of the process.   

To test the hybrid scheme for MET problems,  we consider 
the above birth-death process $\mathbf{X}(t)$ with $k_1 = 1$, 
so that $n_e = k_2$. We compute the mean time of the birth-death 
process starting from $n_e$ to extinction. 
Following the discussion at the end of Section \ref{sec:prelims}, 
the hybrid scheme is considered extinct when the process 
satisfies $\llbracket \mathbf{Z}(t) \rrbracket = 0$. 
We choose the blending functions according to
(\ref{eq:multi_species_blending}), simulating the process 
for different values of $I_1^i$ and $I_2^i$. Since the blending 
region is bounded, we can use the thinning-based 
Algorithm \ref{alg:hybrid3} to simulate the jump-diffusion process 
within the blending region, choosing 
$\Lambda_1 = k_1 I_2^1$ and $\Lambda_2 = k_2$.  
A timestep of $\delta t = \Delta t = 10^{-2}$ is chosen throughout. 

In Figure~\ref{fig:mfpt_bd}, we plot the MET of the hybrid 
scheme for varying {$k_2$}, each point generated from $10^5$ 
independent realisations, and for different choices of blending 
regions. The MET for the CME and CLE, computed directly 
from (\ref{eq:bd_mfpt}) and (\ref{eq:bd_cle_mfpt}), respectively, 
are shown for comparison. It is evident from the numerical 
experiments that the hybrid scheme provides a better 
approximation for the MER compared to the CLE. However, 
the improvement is not uniform over all timescales: the region 
in which the jump process is simulated must be increased 
to correctly capture rare events. Figure \ref{fig:mfpt_bd} 
suggests that the width of the blending region also plays 
a role in the simulation. Indeed, a blending region 
with $(I_1^i, I_2^i) = (3,5)$ appears to be sufficient 
to accurately estimate the MET up to $k_1 = 9$, although 
it is likely this approximation will break down, if
$k_2$ is increased further.  The necessity of tuning 
the blending region to capture the escape time dynamics 
is a disadvantage.  Nonetheless, the hybrid scheme 
provides us with an approach for improving the MET estimate 
obtained from the CLE, at the ``cost'' of having to 
simulate discrete jumps in (increasingly large) 
regions of the domain.  

\begin{figure}[t]
\centerline{
\includegraphics[width=0.46\textwidth]{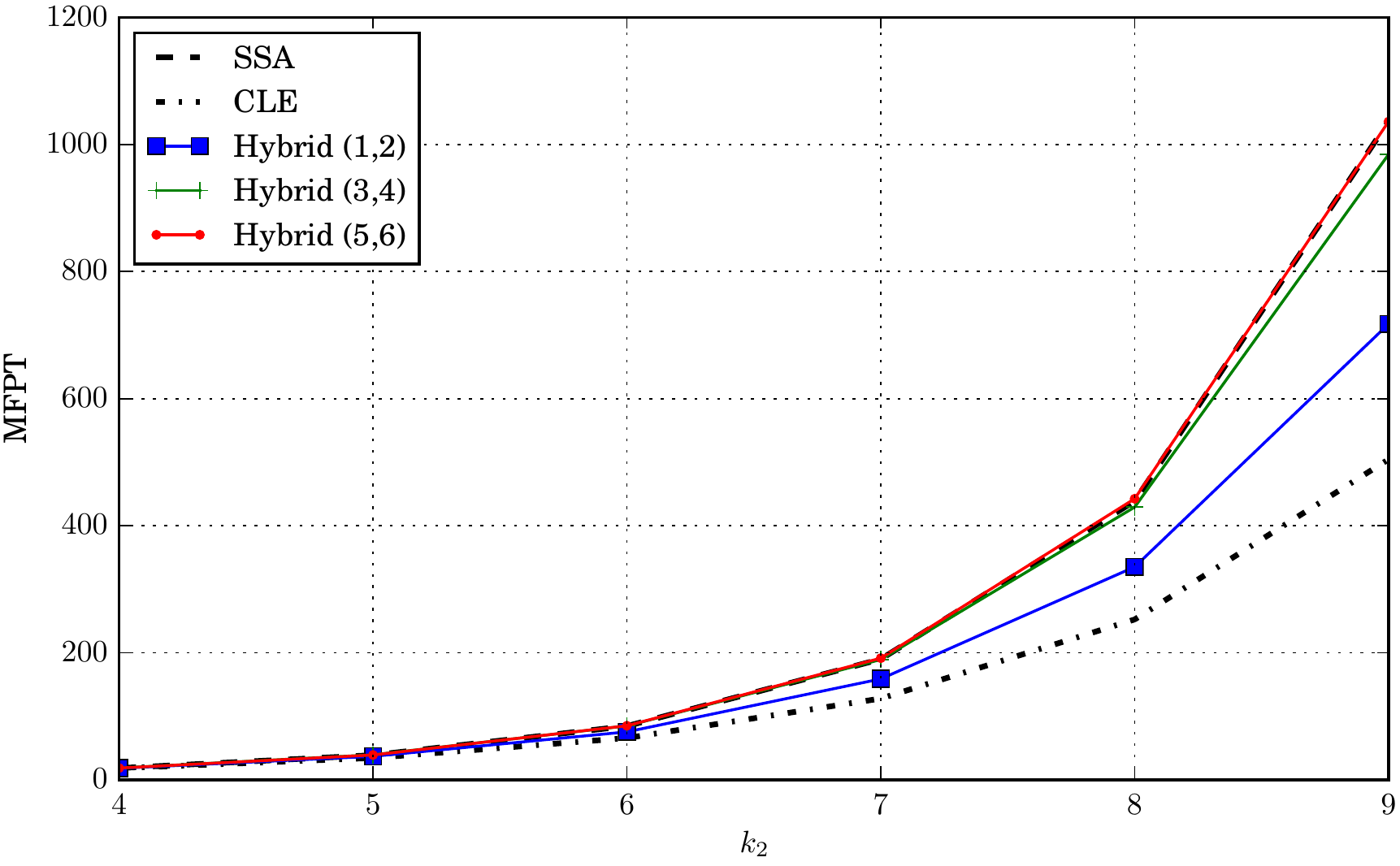}
\hskip 6mm
\includegraphics[width=0.46\textwidth]{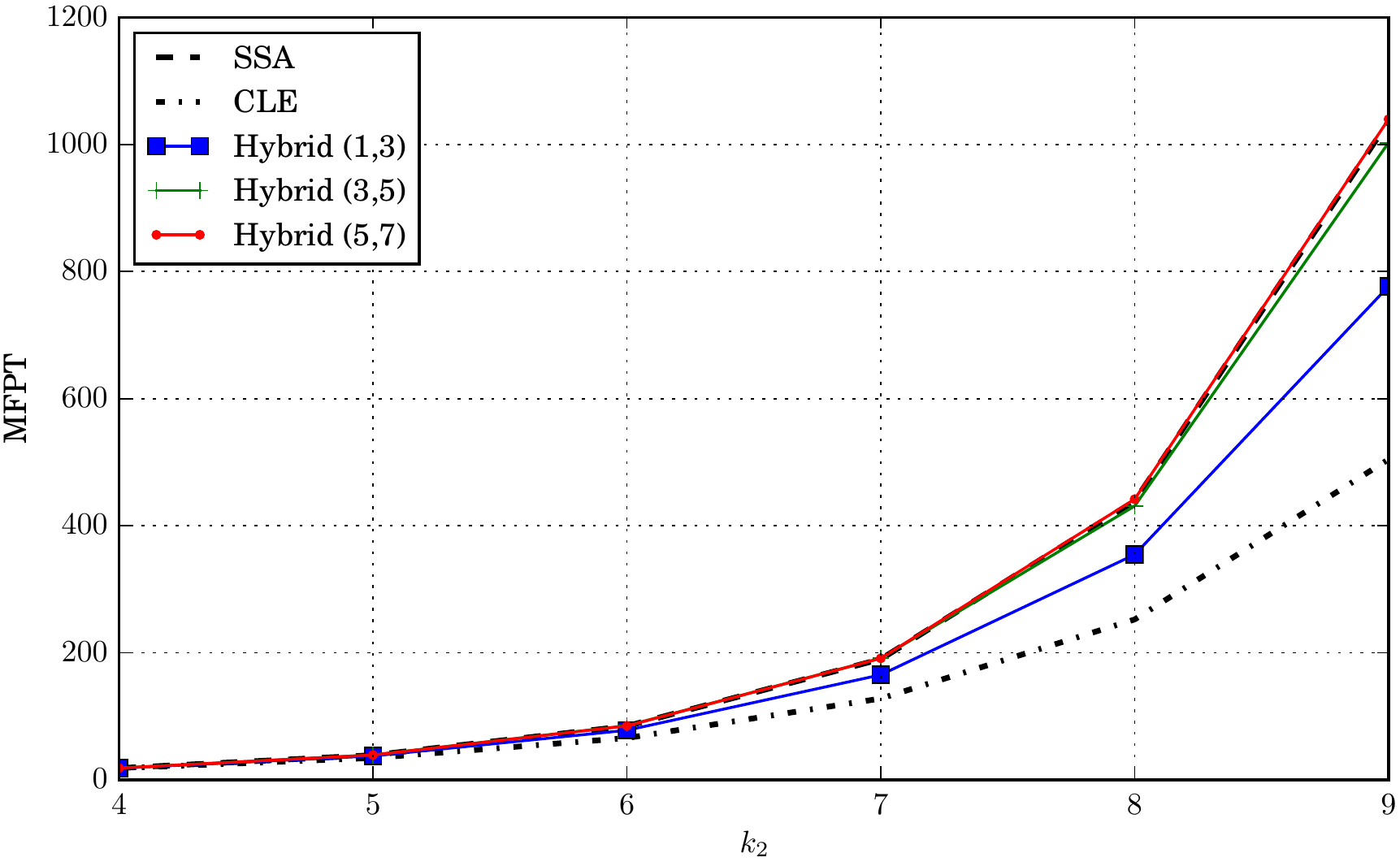}
}
\vskip -4.6cm
\leftline{\hskip 2mm (a) \hskip 6.2cm (b)}
\vskip 4cm
\caption{{\it Comparison of the METs of the birth-death process 
$(\ref{eq:bd})$, the corresponding CLE $(\ref{eq:cle_bd})$ and 
the hybrid scheme. The parameters used are {$k_2 = 4, 5, 6, 7, 8, 9$} 
and $k_1=1$ and the process is started from $A(0)=k_2$. We use:}
(a) {\it blending regions of width $1$;}
(b) {\it blending regions of width $2$.}
}
\label{fig:mfpt_bd}
\end{figure}

\section{Conclusions}
\label{sec:conclusion}

\noindent
In this paper we have introduced a jump-diffusion model for 
simulating multiscale reaction systems efficiently while 
still accounting for the discrete fluctuations where necessary. 
Fast reactions are simulated using the CLE, while the standard 
discrete description is used for slow ones. 
Our approach involves the introduction of a set of blending 
functions~(\ref{eq:multi_species_blending}) which allow one 
to make explicit in which regions the continuum approximation 
should be expected to hold.

Based on the representation of the Markov jump process as a 
time changed Poisson process, we described three different schemes, 
based on \cite{Gillespie:1977:ESS,anderson2007modified} to 
numerically simulate the jump-diffusion model in the three 
different regimes (discrete, continuous and hybrid). 
To demonstrate the efficacy of the schemes, we  simulated 
equilibrium distributions of chemical systems and computed 
extinction times of chemical species for illustrative chemical
systems. The results suggest that the proposed algorithm is 
robust, and is able to handle multiscale processes efficiently 
without the breakdown associated when using the CLE directly.

\section*{Acknowledgements}

\noindent 
The research leading to these results has received funding from 
the European Research Council under the \textit{European Community}'s 
Seventh Framework Programme ({\it FP7/2007-2013})/ERC {\it grant agreement} 
n$^\circ$ 239870. Radek Erban would also like to thank the Royal Society 
for a University Research Fellowship and the Leverhulme Trust for 
a Philip Leverhulme Prize. Andrew Duncan was also supported by 
the EPSRC grant EP/J009636/1.

\end{document}